\documentclass[journal,onecolumn,12pt,twoside]{IEEEtranTCOM}
\usepackage[doublespacing,nodisplayskipstretch]{setspace}
\usepackage[utf8]{inputenc} 
\usepackage[T1]{fontenc}
\usepackage{url,comment}
\usepackage{epsfig,graphics,graphicx}
\usepackage{amsmath,amssymb,amsthm}

\newtheorem{theorem}{Theorem}
\newtheorem{lemma}{Lemma}
\newtheorem{corollary}{Corollary}

\newtheorem{proposition}{Proposition}

\newcommand{\pr}{\mathbb{P}}
\newcommand{\R}{\mathbb{R}}
\newcommand{\Z}{\mathbb{Z}}
\newcommand{\F}{\mathbb{F}}
\newcommand{\hX}{\hat{X}}
\newcommand{\hY}{\hat{Y}}
\newcommand{\hx}{\hat{x}}
\newcommand{\hy}{\hat{y}}
\newcommand{\hU}{\hat{U}}
\newcommand{\hpi}{\hat{\pi}}
\newcommand{\tX}{\tilde{X}}
\newcommand{\tY}{\tilde{Y}}
\newcommand{\tx}{\tilde{x}}
\newcommand{\ty}{\tilde{y}}
\newcommand{\bgamma}{\overline{\gamma}}
\newcommand{\bsigma}{\overline{\sigma}}

\DeclareMathOperator\erfc{erfc}
\allowdisplaybreaks[1]

\begin{document}
%
% paper title
% can use linebreaks \\ within to get better formatting as desired
\title{Time-Entanglement QKD: Secret Key Rates\\[-1.5ex] and  Information Reconciliation Coding}
%in PPM Schemes}
%
%
% author names and IEEE memberships
% note positions of commas and nonbreaking spaces ( ~ ) LaTeX will not break
% a structure at a ~ so this keeps an author's name from being broken across
% two lines.
% use \thanks{} to gain access to the first footnote area
% a separate \thanks must be used for each paragraph as LaTeX2e's \thanks
% was not built to handle multiple paragraphs
%

\author{
        Joseph~J.~Boutros,~\IEEEmembership{Senior Member,~IEEE},
        and~Emina~Soljanin,~\IEEEmembership{Fellow,~IEEE}% <-this % stops a space
\thanks{Joseph J. Boutros is with the Department
of Electrical and Computer Engineering, Texas A\& M University, 23874 Doha, Qatar, e-mail: boutros@ieee.org (see https://www.josephboutros.org).}
\thanks{Emina Soljanin is with the Department
of Electrical and Computer Engineering, Rutgers, the State University of New Jersey, Piscataway, NJ 08854, USA, e-mail: (see https://www.ece.rutgers.edu/emina-soljanin).}
\thanks{This  research  is  based  upon  work  supported  by  the  National  Science  Foundation  under  Grant  \# FET-2007203}}

\markboth{For publication in the IEEE, 30 December 2022}%
{}
\maketitle

%%------------------------------------------------------------------
\begin{abstract}
In time entanglement-based quantum key distribution (QKD), Alice and Bob extract the raw key bits from the (identical) arrival times of entangled photon pairs by time-binning. Each of them individually discretizes time into bins and groups them into frames. They retain only the frames with a single occupied bin. Thus, Alice and Bob can use the position of the occupied bin within a frame to generate random key bits, as in PPM modulation. Because of entanglement, their occupied bins and their keys should be identical. However, practical photon detectors suffer from time jitter errors. These errors cause discrepancies between Alice's and Bob's keys. Alice sends information to Bob through the public channel to reconcile the keys. The amount of information determines the secret key rate. This paper computes the secret key rates possible with detector jitter errors and constructs codes for information reconciliation to approach these rates.
\end{abstract}

%%------------------------------------------------------------------
\section{Introduction}

Secret key distribution protocols establish a shared sequence of bits between two (or more) distant parties, Alice and Bob, in the presence of an eavesdropper, Eve. The key consists of uniformly random independent bits known only to Alice and Bob. Quantum Key Distribution (QKD) starts by communicating quantum states over a quantum channel. The role of the quantum step is to 1) ensure that no eavesdropping goes undetected and 2) provide a source of perfect randomness in the entanglement-based systems. 

There has been a significant effort to provide high key rates over long distances (see recent surveys \cite{Diamanti2016,qkdsurvey21}). QKD schemes based on time-entangled photons have emerged as a promising technique primarily because each entangled photon pair can carry multiple key bits and thus potentially provide a higher secure key rate over long distances \cite{lee2016highrate,Sarihan:19}. 

Time-entanglement-based QKD (TE-QKD) schemes use Spontaneous Parametric Down-Conversion (SPDC) to generate entangled photon pairs according to a Poisson Process. One of the photons goes to Alice, and the other to Bob. Therefore, Alice and Bob ideally detect their photons simultaneously with exponentially distributed photon inter-arrival times. The most common single-photon detectors are Superconducting Nanowire Single-Photon Detectors (SNSPDs), which exhibit properties closest to ideal sensors. They have low dark count rates, meaning they rarely report photon detection without a photon arrival.
Furthermore, they have low detector downtime $d$ and slight detector timing jitter that manifests as Gaussian noise with zero mean and variance $\sigma_d^2$. Unfortunately, these imperfections are non-negligible:  1) detector jitters and dark counts cause disagreements between Alice's and Bob's keys, and 2) the downtime introduces memory within the raw key bits. The secret key rate loss due to the non-ideal properties of these detectors has been studied most recently in \cite{QKD:chengBS22}.

At a high level, there are two main QKD steps. In the first step, Alice and Bob generate  \textit{raw key} bits using a quantum channel. 
Their respective raw keys may disagree at some positions, be partly known to Eve, and may not be uniformly random because of the aforementioned non-ideal detector properties. In the second step, Alice and Bob process the raw key to establish a shared {\it secret key}. They communicate through the public classical channel to reconcile differences between their raw keys, amplify the privacy of the key concerning Eve's knowledge, and compress their sequences to achieve uniform randomness. 
At the end of the protocol, Alice and Bob 1) have identical uniformly random (binary) sequences and 2) are confident the shared sequence is known only to them. Therefore the secret key is private and hard to guess. This paper focuses on the information reconciliation step. 

Alice and Bob obtain correlated streams of bits (raw keys) by detecting the arrival times of their entangled photons. However, they must communicate over a public channel to agree on a key, i.e., reconcile their differences. Here, we consider one-way information reconciliation schemes in which Alice sends information about her sequence to Bob, who uses it to remove the differences between his and Alice's raw keys. After the information reconciliation, Alice and Bob share Alice's initial raw key. However, the shared key is not secret because of the public channel communication. Alice and Bob perform privacy amplification to correct that, establishing secrecy but shortening the key. 
Since Alice and Bob base their secret key generation on correlated photon arrival times, they follow what is known as the source model in Information Theory \cite[Ch.~22.3]{books:GK2011}.
The secrecy capacity for this model when the eavesdropper has access to public communication but does not have correlated prior information
is equal to the mutual information between Alice's 
and Bob's observations (see, e.g., \cite[p.~567]{books:GK2011}). The secrecy capacity is an achievable upper bound on the post-privacy amplification rate.

Alice and Bob generate their secret keys from the correlated random photon arrivals. There are many ways to extract keys from this correlated information. One popular method is similar to Pulse Position Modulation (PPM); see, e.g., \cite{Zhong2015PhotonefficientQK} and references therein. (Some recently proposed adaptive schemes avoid discarding frames with multiple occupied bins \cite{schemes:ZhouW13,schemes:KarimiSW20}.) In PPM, Alice and Bob synchronize their clocks and discretize their timelines into time frames $N$ time bins. 
In PPM, Alice and Bob agree to retain only time frames in which they both detect a single photon arrival and discard all other frames. This single photon is said to occupy a time bin depending on where within the frame it arrives. 
Since photon inter-arrival times follow an exponential distribution, each bin is occupied independently of other bins. Therefore, the number of raw key bits that PPM decoding can extract from each frame equals $\log N$.

This paper focuses on practical photon detectors that suffer from time jitter errors. Since these errors cause discrepancies between Alice's and Bob's keys, Alice must send information to Bob through the public channel to reconcile the keys. The amount of information determines the secret key rate. This paper computes the secret key rates possible with detector jitter errors and constructs codes for information reconciliation to approach these rates.

This paper is organized as follows: Sec.~\ref{sec:notation} introduces notation and lists the paper's main contributions. Sec.~\ref{sec:Model} presents the TE-QKD channel model. Sec.~\ref{sec:ErrorRate} computes the rates of raw key disagreement caused by detection jitter, and Sec.~\ref{sec_soft_output} derives the correlations between Alice's and Bob's raw keys. Sec.~\ref{sec_mut_info} computes achievable information rates and the secrecy capacity of the TE-QKD channel.
Sec.~\ref{sec_codes} proposes and tests several coding schemes for information reconciliation.

%%------------------------------------------------------------------
\section{Notation and Main Contributions \label{sec:notation}}
The number $N$ of bins per time frame could be any positive integer greater than or equal to $2$, our propositions, lemmas, and theorems have no other constraint on $N$. However, our numerical examples are given for $N=2^m$, $m$ integer, $m \ge 1$. The set $\Z_N$ denotes the set of $N$ integers $\{0, 1, \dots, N-1\}$. The notation $\lfloor x \rfloor$, known as the \textit{floor of $x$} for $x \in \R$, is the largest integer smaller than or equal to $x$.

\noindent
Letters such as $X$, $Y$, $\tX$, and $\tY$ denote continuous random variables, while $\hX$ and $\hY$ are discrete random variables. Then, $p(\hy|\hx)$ denotes the conditional probability $\pr(\hY=\hy|\hX=\hx)$.
Also, $p(y|\hx)$ denotes the conditional density $p_{Y|\hX}(y|\hx)$.

\noindent
We use Bourbaki's notation for intervals on the real line, where $a$ and $b$ are two real numbers: 
the closed interval $[a,b]=\{ x \in \R ~:~ a \le x \le b \}$, the half-open intervals $[a,b[=[a,b] \setminus \{b\}$ and $]a,b]=[a,b] \setminus \{a\}$, and the open interval $]a,b[=[a,b] \setminus \{a, b\}$.

\noindent
We use the standard Bachmann-Landau big $\mathcal{O}$ notation:\\
The formal definition of $f(\sigma)=\mathcal{O}(g(\sigma))$ is: 
$\exists \alpha >0$, $\exists \sigma_0 >0$, $\forall \sigma<\sigma_0$, 
$|f(\sigma)|\le \alpha |g(\sigma)|$.\\
In this paper, an expression such as $1-\mathcal{O}(g(\sigma))$
or $1+\mathcal{O}(g(\sigma))$ implicitly assumes that $g(\sigma)>0$ 
in some open interval $]0,\sigma_0[$. Furthermore, we will frequently use $\gamma=1/\sigma^2$, a signal-to-noise ratio defined as the inverse of the jitter variance, then we could write $f(\gamma)=\mathcal{O}(g(\gamma))$ in a similar situation when $\gamma \rightarrow \infty$.\\
Two functions $f:\R \rightarrow \R$ and $g:\R \rightarrow \R$ are asymptotically equivalent if $\lim_{\gamma \rightarrow \infty} \frac{f(\gamma)}{g(\gamma)}=1$. In that case, we write $f(\gamma) \sim g(\gamma)$.

\noindent
The function $Q(x)=\frac{1}{2}\erfc(\frac{x}{\sqrt{2}})=\mathcal{O}(\exp(-x^2/2))$ is the Gaussian tail function.
Recall the definition $Q(x)=\int_x^{\infty} \phi(t) dt$, 
where $\phi(t)=\frac{1}{\sqrt{2\pi}} \exp(-t^2/2)$ is the standard normal density. Furthermore, we recall the binary entropy function, $H_2(x)=-x\log(x)-(1-x)\log(1-x)$, and the symmetric ternary entropy function, 
$H_3(x)=-(1-2x)\log(1-2x)-2x\log(x)$. 

The main contributions of this paper constitute a full characterization of the time-entanglement QKD channel, from information theory and coding theory point of view:
\begin{itemize}
\item We derive the error rates of the TE-QKD channel, and prove that the TE-QKD channel behaves like an $1/2$-diversity Nakagami fading channel, see Proposition~\ref{prop_pe_uncoded}.
\item We find the exact {\em a priori} probability of bins given that both Alice's and Bob's frames are valid, see Lemma~\ref{lem_U_tX}.
\item We establish the exact conditional density of Bob's photon position given Alice's photon bin, for a soft-output TE-QKD channel, see Theorem~\ref{th_Y_cond_hX_tY}. The output density expression is also  determined, see~(\ref{equ_pdfY_valid}). 
\item We determine the expression of the transition probabilities of the discrete (hard-output) TE-QKD channel, see Corollary~\ref{cor_pij}.
\item We give the exact expression of the {\em a posteriori} probability for the soft-output TE-QKD channel, see Theorem~\ref{th_app}.
\item We derive the exact formula for the mutual information $I(\hX;\hY)$ (hard-output) and find simplified expressions in the small-noise regime, see (\ref{equ_I_hX_hY}), (\ref{equ_approxI_discrete}), (\ref{equ_approx_discrete_circular}), and Proposition~\ref{prop_high_snr}-c.
\item The exact formula for the mutual information $I(\hX;Y)$ (soft-output) is given, see (\ref{equ_I_hX_Y}). We also determine all densities needed to compute the maximal rate $I(X;Y)$ and we give
a nice log-formula expression in the small-noise regime, see Theorem~\ref{th_capa} and Corollary~\ref{cor_log_formula}.
\item The last section, Section~\ref{sec_codes}, shows new results with huge coding gains obtained by short and moderate-length error-correcting codes such as RS, BCH, and LDPC codes under algebraic hard-decision decoding and probabilistic soft-decision decoding.  
\end{itemize}
%%------------------------------------------------------------------
\section{PPM Channel Model \label{sec:Model}}
Let $\tX$ and $\tY$ represent the time-position of the received photons at Alice's and Bob's sides, respectively. An illustration of this QKD scheme is given in Figure~\ref{fig_qkd_scheme}. 
\begin{figure}[hbt]
   \centering
   \includegraphics[width=12cm]{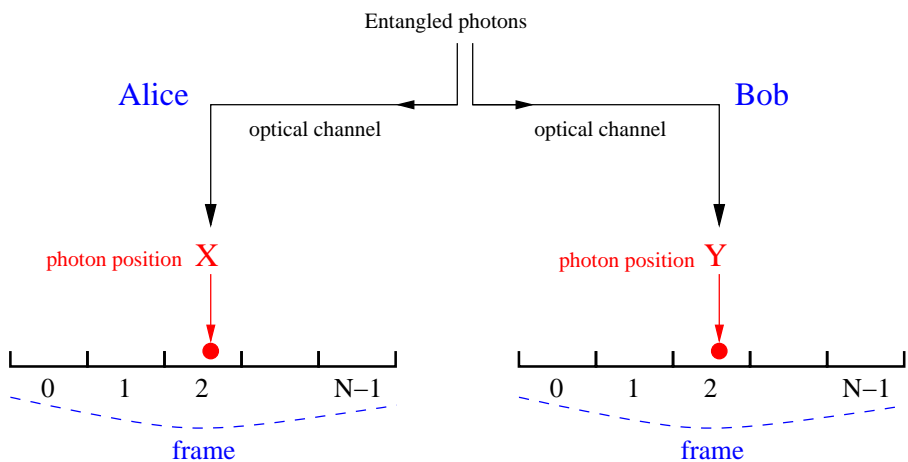}
   \caption{QKD based on time entanglement with $N$ bins per frame, $\log_2(N)$ binary digits per bin.}
\label{fig_qkd_scheme}
\end{figure}

We adopt the following mathematical model for the positions of two time-entangled photons: \begin{equation}
\label{equ_XYmodel}
\tX=U+Z_1,~~~\tY=U+Z_2,
\end{equation}
where $Z_1$ and $Z_2$ are independent identically distributed
$\mathcal{N}(0,\sigma^2)$ additive Gaussian noises modeling the detection jitter. $U$ is a real uniform random variable
in the interval $[0, N[$, where the integer $N=2^m$ is the number of bins per frame, and $m$ is the number of bits per photon. Alice and Bob communicate via a public channel and agree on a valid frame when $\tX$ and $\tY$ fall in the interval $[0, N[$. They reject empty frames and frames with more than one received photon. Under the model defined by (\ref{equ_XYmodel}), the probability of a frame 
to be valid for both Alice and Bob is $\pr(\tX,\tY \in [0,N[)$. 
Let $X$ and $Y$ denote the instances of $\tX$ and $\tY$ within the interval $[0, N[$, and let $\hX$ and $\hY$ be the bin number inside a frame, i.e.,
\begin{align}
 \tX=X,~~\text{for}~\tX \in [0,N[, & ~~~~~~\tY=Y,~~\text{for}~\tY \in [0,N[,\\
\hX=\lfloor X \rfloor \in \Z_N, & ~~~~~~
\hY=\lfloor Y \rfloor \in \Z_N.
\end{align}
From an information theoretical perspective, we distinguish two communication channels between Alice and Bob: (a) an algebraic (hard) output channel, (b) a real (soft) output channel, both having a discrete $N$-ary input $\hX$ as shown in Figure~\ref{fig_inf_theo_channels}. 

\begin{figure}[hbt]
   \centering
   \includegraphics[width=13cm]{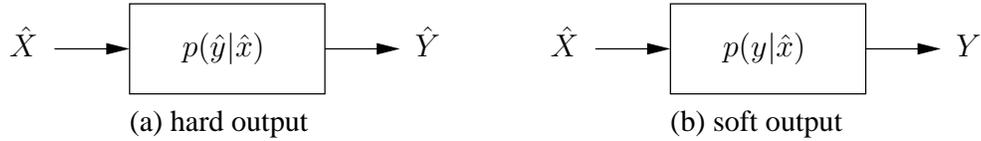}
   \caption{Channel models for hard-decision decoding (a) and soft-decision decoding (b).}
\label{fig_inf_theo_channels}
\end{figure}

Without error-correcting codes, the information rate on these channels 
is $\log_2(N)=m$ bits per channel use (bpcu). 
The main channel parameter $\gamma$ is a signal-to-noise ratio parameter (SNR) defined as
\begin{equation}
\gamma=\frac{E_s}{\sigma^2}=\frac{1}{\sigma^2}, 
\end{equation}
where the average energy per symbol $E_s=1$ is a normalized energy cost 
per transmitted photon. Another QKD channel parameter is $\bgamma$, referred 
to as the normalized signal-to-noise ratio, where the standard deviation of the additive Gaussian noise is normalized by the frame length $N$, hence its definition is
\begin{equation}
\bgamma=\frac{1}{(\sigma/N)^2}=\frac{N^2}{\sigma^2}, 
~~~\bgamma(dB)=\gamma(dB)+20\log_{10}(N).  
\end{equation}

We express the probability of error and the information rate as functions of $N$ and the SNR $\gamma$ or the normalized SNR $\bgamma$. The bin width within a frame is set to $1$ to simplify the analysis, i.e., the frame width is $N$ in all sections except for Section~\ref{sec_capacity}. The conversion of this mathematical model into a physical model representing a laboratory experiment is straightforward after introducing a time scale to convert $\gamma$ and $N$ into physical parameters. In Section~\ref{sec_capacity}, the number of bins is infinite (it's a continuum of bins), the frame has a unit length and $\gamma=\bgamma$ in that special QKD channel with both soft input and soft output. 

%%------------------------------------------------------------------
\section{Rate of Raw Key Disagreement Under Detection Jitter \label{sec:ErrorRate}}
We consider the probability of error $P_e(\gamma)=\pr(\hX \ne \hY)$. 
The probability $P_e$ characterizes the quality of channel~(a) in Figure~\ref{fig_inf_theo_channels} defined by its transition probabilities $p(\hy|\hx)$.
The latter will be entirely determined in Section~\ref{sec_mut_info}. In the current section, 
we are interested in determining the expression of $P_e(\gamma)$ as a function of the signal-to-noise ratio $\gamma$, for a given number of bins $N$ per frame.\\

Let $\pi_i=\pr(\hX=i)$, $i \in \Z_N$, be the {\em a priori} probability of the unique frame photon to fall in bin number $i$. Then, the exact expression of the probability of error is
\begin{equation}
\label{equ_pe_1}
P_e(\gamma)=\sum_{i=0}^{N-1} \pi_i \sum_{\substack{j=0\\j \ne i}}^{N-1} p(\hy=j|\hx=i)
=\frac{1}{N} \sum_{i=0}^{N-1} \pr(\hY \ne \hX | \hU=i),
\end{equation}
where $\hU=\lfloor U \rfloor$. Since $U$ is uniform in $[0,N)$, 
we get $\pr(\hU=i)=\pr(U \in [i,i+1))=\frac{1}{N}$ which explains the factor in the last equality above. As a first step, in the current section, we solve $P_e(\gamma)$ from the most right equality in~(\ref{equ_pe_1}) via the conditioning over $\hU$. To avoid cumbersome expressions, exact expressions as established in Sections~\ref{sec_soft_output}\&\ref{sec_mut_info}, we assume that $\gamma$ is large enough ($\sigma^2$ is small enough) so we can neglect the border effects in the frame. Hence, we make no difference here between $\tX$ and $X$ (resp. $\tY$ and $Y$), 
and we use the approximation that both $X$ and $Y$ are i.i.d. Gaussian when conditioning on $U$. 

\begin{proposition}
\label{prop_pe_uncoded}
The probability of symbol error $P_e(\gamma)=\pr(\hX \ne \hY)$ 
as a function of the SNR $\gamma$ and the number $N$ of bins per frame is given by the expression
\begin{equation}
\label{equ_pe_2}
P_e(\gamma)=\frac{2}{\sqrt{\pi}} \times \left( 1-\frac{1}{N}\right) \times \gamma^{-\frac{1}{2}}
+ \mathcal{O}(\exp(-\frac{\gamma}{4})).     
\end{equation}
\end{proposition}

\begin{proof}
Set $V=U-\hU$, so $V$ is $\text{Uniform}[0,1]$.
Let $p(i \rightarrow j|v)$ be the probability of falling in bin $j$ given that $\hU=i$ and $V=v$, where $i,j \in \Z_N$.
A symbol error occurs if $X=U+Z_1$ remains in bin $i$ but
$Y=U+Z_2$ leaves to bin $j$, $j \ne i$. The probability of such an event is $p(i \rightarrow i|v) \times p(i \rightarrow j|v)$, given that both additive Gaussian noises $Z_1$ and $Z_2$ are independent. 
Also, an error occurs if both $X$ and $Y$ leaves to two different bins $\ell$ and $j$,
with probability $p(i \rightarrow \ell|v) \times p(i \rightarrow j|v)$. 
Then, the conditional symbol error probability becomes
\[ 
P_e(i, v) = 2 \left[ \sum_{\substack{j=0\\ j\ne i}}^{N-1}
p(i \rightarrow i|v) p(i \rightarrow j|v) + 
\sum_{\substack{\ell=0\\ \ell\ne i}}^{N-1} 
\sum_{\substack{j=0\\ j\ne i, j\ne \ell}}^{N-1} p(i \rightarrow \ell|v) p(i \rightarrow j|v) \right].
\]
The factor of $2$ is due to the symmetry if the two letters $X$ and $Y$ are switched. 
As illustrated in Figure~\ref{fig_uncoded_pes}, we will neglect bins beyond the left and the right bin. The neglected bins are at least at distance $1.0$ from the bin $\hU=i$. They correspond to a probability of error $Q(1/\sigma)=\mathcal{O}(\exp(-1/(2\sigma^2)))=\mathcal{O}(\exp(-\frac{\gamma}{2}))$.
To further simplify the notations, define $p_1$, $p_2$, and $p_3$,
where
\begin{equation}
\label{equ_p1_prop1}
p_1=p(i \rightarrow i|v)=1-Q\left(\frac{v}{\sigma}\right)-Q\left(\frac{1-v}{\sigma}\right),
\end{equation}
\begin{equation}
\label{equ_p2_prop1}
p_2=p(i \rightarrow i-1|v)=Q\left(\frac{v}{\sigma}\right),
\end{equation}
and 
\begin{equation}
\label{equ_p3_prop1}
p_3=p(i \rightarrow i+1|v)=Q\left(\frac{1-v}{\sigma}\right),
\end{equation}

we obtain
\[
P_e(i, v) = \mathcal{O}(\exp(-\frac{\gamma}{2})) ~+~
\left\lbrace
\begin{array}{l}
2[p_1p_2+p_1p_3+p_2p_3],~~\text{for}~i=1\dots N-1,\\
2p_1p_3,~~\text{for}~i=0,\\
2p_1p_2,~~\text{for}~i=N-1.
\end{array}
\right.
\]
\begin{figure}[hbt]
   \centering
   \includegraphics[width=12cm]{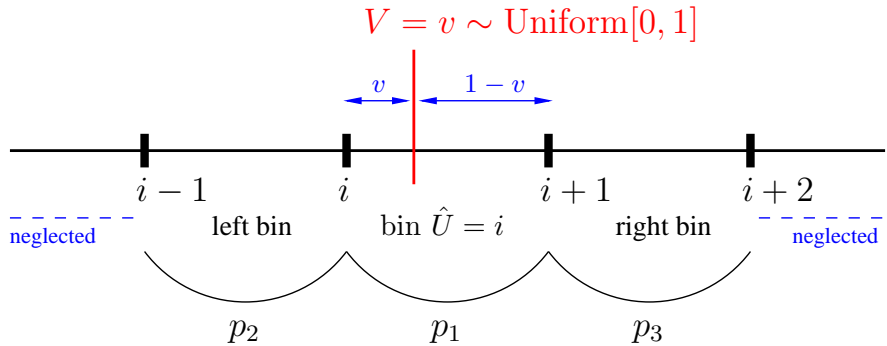}
   \caption{Illustration of the probability of error in bin position.}
\label{fig_uncoded_pes}
\end{figure}

Now, integrate over $v$, 
\[
\pr(\hY \ne \hX | \hU=i) = \int_0^1 P_e(i,v) ~dv.
\]
Then, apply (\ref{equ_pe_1}) and use $\int_0^1p_1p_2dv=\int_0^1 p_1p_3dv$ to finally reach
\begin{align}
P_e(\gamma) &=\frac{1}{N} \sum_{i=0}^{N-1} \int_0^1 P_e(i,v) ~dv \\
            &=4\left(1-\frac{1}{N}\right)\int_0^1 p_1p_2 ~dv + 2\left(1-\frac{1}{N}\right)\int_0^1 p_2p_3 ~dv ~+  \mathcal{O}(\exp(-\frac{\gamma}{2})).\label{equ_pe_p1p2_p2p3}      
\end{align}
The two integrals in (\ref{equ_pe_p1p2_p2p3}) include three types of integrals.
Let us process them step by step.
\[
I_1=\int_0^1 Q\left(\frac{v}{\sigma}\right)~dv = \frac{\sigma (1-e^{-1/2\sigma^2})}{\sqrt{2\pi}}+Q\left(\frac{1}{\sigma}\right)~=~\frac{\sigma}{\sqrt{2\pi}}+\mathcal{O}(\exp(-\frac{\gamma}{2})). 
\]
\begin{align*}
I_2 =\int_0^1 \left[ Q\left(\frac{v}{\sigma}\right) \right]^2 dv &= 
\frac{2\sqrt{2}\sigma-2\sigma(1-2Q(\sqrt{2}/\sigma))+4\sqrt{\pi}Q^2(1/\sigma)-4\sqrt{2}\sigma e^{-1/2\sigma^2}Q(1/\sigma)}{4\sqrt{\pi}}\\
&=\frac{(\sqrt{2}-1)\sigma}{2\sqrt{\pi}}+\mathcal{O}(\exp(-\gamma)). 
\end{align*}
\[
I_3=\int_0^1 Q\left(\frac{v}{\sigma}\right)~Q\left(\frac{1-v}{\sigma}\right)~dv \le \int_0^1 \exp\left( \frac{-v^2-(1-v)^2}{2\sigma^2}\right)~dv ~=~\mathcal{O}(\exp(-\frac{\gamma}{4})),
\]
since $v^2+(1-v)^2\ge \frac{1}{2}$ for $v \in [0,1]$. $I_1$ and $I_2$ were solved via integration by parts using the fact that $\frac{dQ(x)}{dx}=-\phi(x)$. 
$I_3$ has no simpler form. 
Our upper bound of $I_3$ brings a sufficient answer to the current proposition. 
After substituting $I_1$, $I_2$, and $I_3$ into (\ref{equ_pe_p1p2_p2p3}), we 
get (\ref{equ_pe_2}) as stated by the proposition, where $\sigma=\gamma^{-\frac{1}{2}}$. 
\end{proof}

The expression 
$\frac{2}{\sqrt{\pi}} \times \left( 1-\frac{1}{N}\right) \times \gamma^{-\frac{1}{2}}$ perfectly fits the Monte Carlo simulation
of $\pr(\hY \ne \hX)$ even for a signal-to-noise ratio as low as 20dB (error rate close to $10^{-1}$). 
Figure~\ref{fig_perf_uncoded} shows the plots of the probability of error $P_e(\gamma)$ for different number of bins per frames, from 1 bit per photon up to 4 bits per photon.
The plots of the probability of error versus the normalized SNR, $P_e(\bgamma)$,
are obtained from Figure~\ref{fig_perf_uncoded} after shifting right each curve by $20\log_{10}(N)$~decibels. 

\begin{figure}[hbt]
   \centering
   \includegraphics[width=9cm,angle=270]{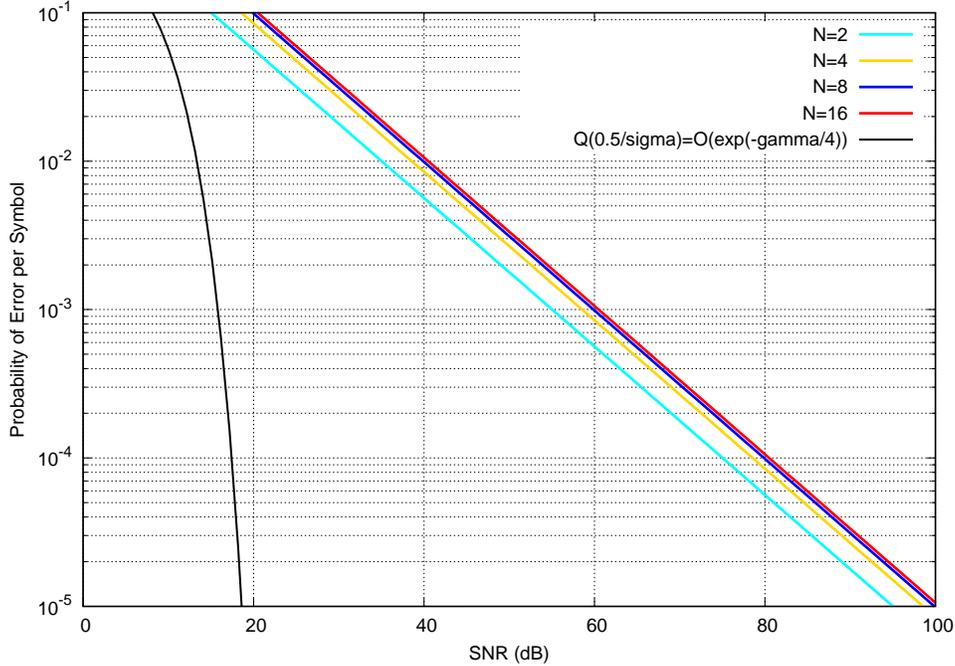}
   \caption{Probability of symbol error versus SNR, $\log_2(N)$ bits per photon, no coding.}
\label{fig_perf_uncoded}
\end{figure}

%%------------------------------------------------------------------
\section{Correlation Between Raw keys \label{sec_soft_output}}
The conditional densities of $\tX$ is directly derived from (\ref{equ_XYmodel}), 
\begin{equation}
\label{equ_condXt_U}
p_{\tX|U}(\tx|u) = \frac{1}{\sqrt{2\pi\sigma^2}} \exp\left(- \frac{(\tx-u)^2}{2\sigma^2}\right),~~~\tx \in \R. 
\end{equation}
Conditioned on $U=u$, $\tX$ and $\tY$ are independent. Then, after integrating (\ref{equ_condXt_U}), 
\[
\pr(\tX,\tY \in [0,N[|u)=\pr(\tX \in [0,N[|u)^2
=\left[ Q\left(-\frac{u}{\sigma}\right)-Q\left(\frac{N-u}{\sigma}\right) \right]^2. 
\]
So the probability of both Alice's and Bob's frames are valid is
\begin{align}
\pr(\tX,\tY \in [0,N[) &= 
\int_0^N  \left[ Q\left(-\frac{u}{\sigma}\right)-Q\left(\frac{N-u}{\sigma}\right) \right]^2~p_U(u)\, du \nonumber \\
&= \frac{1}{N}\int_0^N 
\left[ Q\left(-\frac{u}{\sigma}\right)-Q\left(\frac{N-u}{\sigma}\right) \right]^2\, du. \label{equ_tXtY_0N}
\end{align}
The density of $\tX$ is also derived by integrating over $u$, 
which is equivalent to convolving the densities of $U$ and $Z_1$, we get 
\begin{equation}
\label{equ_pdfXt}
p_{\tX}(\tx)=\int_0^N p_{\tX|U}(\tx|u) \cdot \frac{1}{N}~du = 
\frac{1}{N} \left[ Q\left(\frac{-\tx}{\sigma}\right)-Q\left(\frac{N-\tx}{\sigma}\right)\right],~~~\tx \in \R.
\end{equation}
Since $X$ is a version of $\tX$ truncated to the interval $[0,N[$, 
conditioning on $U+Z_1 \in [0,N[$, the density of $X$ is determined by scaling the density of $\tX$, namely 
\begin{equation}
\label{equ_pdfX_condU}
 p_{X|U}(x|u)=\frac{p_{\tX|U}(x|u)}{\int_0^N p_{\tX|U}(t|u) ~dt}
 =\frac{\frac{1}{\sqrt{2\pi\sigma^2}}
 \exp\left(-\frac{(x-u)^2}{2\sigma^2}\right)}
    {\left[ Q\left(-\frac{u}{\sigma}\right)-Q\left(\frac{N-u}{\sigma}\right) \right]},~~~~x,u \in [0,N[,  
\end{equation}
and
\begin{equation}
\label{equ_pdfX}
 p_X(x)=\frac{p_{\tX}(x)}{\int_0^N p_{\tX}(t) ~dt}
 =\frac{Q\left(-\frac{x}{\sigma}\right)-Q\left(\frac{N-x}{\sigma}\right)}
    {\int_0^N \left[ Q\left(-\frac{t}{\sigma}\right)-Q\left(\frac{N-t}{\sigma}\right) \right]~dt},~~~~x \in [0,N[. 
\end{equation}
By symmetry from (\ref{equ_XYmodel}), $p_{\tY|U}(\ty|u)$,  $p_{\tY}(\ty)$, $p_{Y|U}(y|u)$, and $p_Y(y)$ have expressions identical to (\ref{equ_condXt_U}), (\ref{equ_pdfXt}), (\ref{equ_pdfX_condU}), and~(\ref{equ_pdfX}) respectively, for $\ty \in \R$ and $y \in [0,N[$. 
The bins {\em a priori} probabilities $\pi_i=\pr(\hX=i)=\pr(X \in [i,i+1])$ become,
\begin{equation}
\label{equ_pi_i} 
    \pi_i=\pr(\hX=i)=\int_i^{i+1} p_X(x) ~dx = \frac{\int_i^{i+1} \left[Q\left(\frac{-x}{\sigma}\right)-Q\left(\frac{N-x}{\sigma}\right)\right] ~dx}
    {\int_0^N \left[ Q\left(-\frac{t}{\sigma}\right)-Q\left(\frac{N-t}{\sigma}\right) \right]~dt},~~~i \in \Z_N.
\end{equation}
At high SNR, for $\sigma^2 \ll 1$, we have $\pi_i \approx 1/N$, $\forall i$, because the truncation to the interval $[0,N[$ has less effect in the small-noise regime. Numerical examples are given in Table~\ref{table_pis}, for $N=8$ bins per frame. The entropy of $\hX$ is very stable, as listed in the last column of the table,
$H(\hX)=-\sum_{i=0}^{N-1} \pi_i \log_2(\pi_i) \approx \log_2(N)$ at low and high signal-to-noise ratios. 

\begin{table}[hbt]
\begin{center}
\caption{\label{table_pis} A priori probabilities of photon bins for $N=8$ bins per frame.}
\begin{tabular}{|c|c|c|}
\hline
SNR & $\pi_0, \dots, \pi_7$ & $H(\hX)$ (bits) \\ \hline \hline 
10~dB & {\small 0.112796, 0.129062, 0.129071, 0.129071, 0.129071, 0.129071, 0.129062, 0.112796} & 2.997655 \\ \hline
25~dB & {\small 0.122885, 0.125705, 0.125705, 0.125705, 0.125705, 0.125705, 0.125705, 0.122885} & 2.999931 \\ \hline 
40~dB & {\small 0.124626, 0.125125, 0.125125, 0.125125, 0.125125, 0.125125, 0.125125, 0.124626} & 2.999998 \\ \hline
\end{tabular}
\end{center}
\end{table}

The following lemma helps understand the analytic behavior of (\ref{equ_tXtY_0N})-(\ref{equ_pi_i}) at high SNR,
when $\sigma^2 \ll 1$.

\begin{lemma}
\label{lem_Q_minus_Q}
Let $f_{\sigma}(x)=Q\left(-\frac{x}{\sigma}\right)-Q\left(\frac{1-x}{\sigma}\right)$.
For $\sigma>0$ and $\gamma=1/\sigma^2$, 
given the properties of the Gaussian tail function $Q(x)$, the difference function $f_{\sigma}(x)$ satisfies\\
a) $\forall x \in \R, f_{\sigma}(x)=f_{\sigma}(1-x) \in ]0,1[$. Also, $f_{\sigma}(0)=f_{\sigma}(1)=\frac{1}{2}-\mathcal{O}(\exp(-\frac{\gamma}{2}))$.\\
b) For $x \in ]0,1[$, $f_{\sigma}(x)=1-\mathcal{O}(\exp(-\min^2(x,1-x)\cdot\frac{\gamma}{2}))$.\\
c) For $x<0$, we have $f_{\sigma}(x)=\mathcal{O}(\exp(-x^2\cdot\frac{\gamma}{2}))$,
and $f_{\sigma}(x)=\mathcal{O}(\exp(-(x-1)^2\cdot\frac{\gamma}{2}))$ for $x>1$.\\
d) Integrating $f_{\sigma}$ and $f_{\sigma}^2$, we get $\int_0^1 f_{\sigma}(x)~dx ~=~ 1- \sqrt{\frac{2}{\pi}}\cdot \frac{1}{\sqrt{\gamma}} + \mathcal{O}(\exp(-\frac{\gamma}{2})) ~=~1- \mathcal{O}(\frac{1}{\sqrt{\gamma}})$ 
and $\int_0^1 f_{\sigma}^2(x)~dx ~=~ 1- {\frac{1+\sqrt{2}}{\sqrt{\pi}}}\cdot \frac{1}{\sqrt{\gamma}} + \mathcal{O}(\exp(-\frac{\gamma}{4}))~=~1- \mathcal{O}(\frac{1}{\sqrt{\gamma}})$.\\
e) $\int_{i/N}^{(i+1)/N} f_{\sigma}(x)~dx=\frac{1}{N}-\mathcal{O}(\frac{1}{\sqrt{\gamma}})$ for $i=0$ and $i=N-1$ (the two extreme bins in a frame of $N$ bins).  
$\int_{i/N}^{(i+1)/N} f_{\sigma}(x)~dx=\frac{1}{N}+\mathcal{O}(\exp(-\beta\gamma))$ for $i=1 \ldots N-2$ (the inner bins), where the exponent constant is $\beta=\frac{1}{2}\min^2(\frac{i}{N}, 1-\frac{i+1}{N})$.
\end{lemma}
\begin{proof}
For a), let $G$ be a standard normal random variable.
The finite interval $[-x,1-x]$ is never reduced to a single point. 
We get $f_{\sigma}(x)=\pr(G \in [-x,1-x]) \in ]0,1[$. 
Then, $f_{\sigma}(1-x)=Q\left(-\frac{(1-x)}{\sigma}\right) - Q\left(\frac{x}{\sigma}\right)=1-Q\left(\frac{1-x}{\sigma}\right)-1+Q\left(-\frac{x}{\sigma}\right)=f_{\sigma}(x)$, using the property $Q(-x)=1-Q(x)$. Finally $f_{\sigma}(0)=Q(0)-Q\left(\frac{1}{\sigma}\right)=\frac{1}{2}-\mathcal{O}(\exp(-\frac{\gamma}{2}))$.\\
For b), we write $f_{\sigma}(x)=1-Q\left(\frac{x}{\sigma}\right) - Q\left(\frac{1-x}{\sigma}\right)$. Then $Q\left(\frac{x}{\sigma}\right) + Q\left(\frac{1-x}{\sigma}\right) \le \frac{1}{2}\exp(-x^2/(2\sigma^2))+\frac{1}{2}\exp(-(1-x)^2/(2\sigma^2))\le \exp(-\min^2(x,1-x)\gamma/2)$ which yields the announced result. This inequality is only useful to us for $x\in ]0,1[$ to keep the exponential decay.\\
For c), $x<0$, so $1-x>-x>0$. 
Then $f_{\sigma}(x)\le Q\left(\frac{-x}{\sigma}\right) \le \frac{1}{2}\exp(-x^2/(2\sigma^2))=\mathcal{O}(\exp(-x^2\gamma/2))$. 
The proof is similar for $x>1$.\\
As mentioned for $I_1$ in the proof of Proposition~\ref{prop_pe_uncoded}, the anti-derivative of $Q(ax)$, $a,x \in \R$, is determined after integration by parts. We get 
\begin{equation}
\label{equ_antider}
\int Q(ax) dx = xQ(ax)-\frac{1}{\sqrt{2\pi a^2}}\exp(-a^2x^2/2)+c,
~~\text{where $c$ is the integration constant.}
\end{equation}
For d), $\int_0^1 f_{\sigma}(x)~dx = \int_0^1 \left[ 1-Q\left(\frac{x}{\sigma}\right)-Q\left(\frac{1-x}{\sigma}\right)\right] dx = 1-2I_1 = 1- \sqrt{\frac{2}{\pi}}\cdot \frac{1}{\sqrt{\gamma}} + \mathcal{O}(\exp(-\frac{\gamma}{2}))$,
where $I_1$ is solved thanks to (\ref{equ_antider}).\\  
As mentioned for $I_2$ in the proof of Proposition~\ref{prop_pe_uncoded}, the anti-derivative of $[Q(ax)]^2$, $a,x \in \R$, is also determined by integration by parts and the application of (\ref{equ_antider}). We get 
\begin{equation}
\label{equ_antider2}
\int Q^2(ax) dx = xQ^2(ax)-\sqrt{\frac{2}{\pi a^2}}Q(ax)\exp(-a^2x^2/2)
+\frac{1}{\sqrt{\pi a^2}}Q(ax\sqrt{2})+c.
\end{equation}
Then, $\int_0^1 f_{\sigma}^2(x) dx = 1-4I_1+2I_2+2I_3=1-4\times \frac{\sigma}{\sqrt{2\pi}}+2 \times \frac{\sqrt{2}-1}{2\sqrt{\pi}} \sigma + \mathcal{O}(\exp(-\frac{\gamma}{4}))$, where $I_2$ is solved thanks to (\ref{equ_antider2}) and $I_3=\mathcal{O}(\exp(-\frac{\gamma}{4}))$ as shown before. This completes the proof of d).\\
The proof of e) is mainly based on (\ref{equ_antider}), after taking care of the bin position within the frame.
We have 
\begin{align*}
I_4 &=\int_{i/N}^{(i+1)/N} f_{\sigma}(x)~dx=\int_{i/N}^{(i+1)/N} \left[ 1-Q\left(\frac{x}{\sigma}\right)-Q\left(\frac{1-x}{\sigma}\right)\right] dx\\ 
   &= \frac{1}{N}-\int_{i/N}^{(i+1)/N} Q\left(\frac{x}{\sigma}\right) dx - \int_{1-(i+1)/N}^{1-i/N} Q\left(\frac{x}{\sigma}\right) dx\\
   &= \frac{1}{N}-\left[ \frac{(i+1)}{N} Q\left(\frac{i+1}{N\sigma}\right)
   -\frac{i}{N} Q\left(\frac{i}{N\sigma}\right)
   -\frac{\sigma}{\sqrt{2\pi}} e^{—\frac{\gamma}{2}\frac{(i+1)^2}{N^2}}
   +\frac{\sigma}{\sqrt{2\pi}} e^{—\frac{\gamma}{2}\frac{i^2}{N^2}}
   \right] \\
   &-\left[ (1-\frac{i}{N}) Q\left(\frac{1-\frac{i}{N}}{\sigma}\right)
   -(1-\frac{i+1}{N}) Q\left(\frac{1-\frac{i+1}{N}}{\sigma}\right)
   -\frac{\sigma}{\sqrt{2\pi}} e^{—\frac{\gamma}{2}(1-\frac{i}{N})^2}
   +\frac{\sigma}{\sqrt{2\pi}} e^{—\frac{\gamma}{2}(1-\frac{i+1}{N})^2}
   \right].
\end{align*}
If $i=0$ or $i=N-1$, $I_4=\frac{1}{N}-\frac{\sigma}{\sqrt{2\pi}}=1-\mathcal{O}(\frac{1}{\sqrt{\gamma}})$, all terms with exponential decay are absorbed by the $\mathcal{O}(\frac{1}{\sqrt{\gamma}})$. 
For middle bins, $i=1 \ldots N-2$, $I_4=\frac{1}{N}+\mathcal{O}(e^{-\frac{\gamma}{2}\frac{i^2}{N^2}})+\mathcal{O}(e^{-\frac{\gamma}{2}(1-\frac{i+1}{N})^2})$, all terms of higher decay are absorbed by these two big $\mathcal{O}$. Hence, $I_4=\frac{1}{N}+\mathcal{O}(e^{-\beta\gamma})$, where 
the exponent constant is $\beta=\frac{1}{2}\min^2(i/N, 1-(i+1)/N)$.
\end{proof}
The convergence of $f_{\sigma}(x)$ is not uniform in the interval $[0,1]$. The point-wise convergence of $f_{\sigma}(x)$ to~$0$ (outside $[0,1]$) or to~$1$ (inside $[0,1]$) is very slow in the neighborhood of the points $x=0$ and $x=1$. At high SNR, the difference of the two $Q()$ functions behaves as a square function and its integral slowly approaches $1$ at a rate of $1/\sqrt{\gamma}$. 

Applying Lemma~\ref{lem_Q_minus_Q} to (\ref{equ_tXtY_0N})-(\ref{equ_pi_i}), after substituting $u/N$ to $u$ and $\sigma/N$ to $\sigma$, proves the following equalities
where $\tx, x, u \in ]0,N[$ and $i \in \Z_N$:  
\begin{align*}
\pr(\tX,\tY \in [0,N[) &= 1-\mathcal{O}(\frac{1}{\sqrt{N^2\cdot\gamma}})=1-\mathcal{O}(\frac{1}{\sqrt{\bgamma}}),\\
p_{\tX}(\tx) &= \frac{1}{N}-\mathcal{O}(\exp(-{\min}^2(\frac{\tx}{N},1-\frac{\tx}{N})\cdot\frac{\bgamma}{2})),\\
p_{X|U}(x|u) &= \frac{1}{\sqrt{2\pi\sigma^2}} \exp\left(-\frac{(x-u)^2}{2\sigma^2}\right)
\cdot (1+\mathcal{O}(\exp(-{\min}^2(\frac{u}{N},1-\frac{u}{N})\cdot\frac{\bgamma}{2})),\\
p_X(x) &= \frac{1}{N}\cdot(1-\mathcal{O}(\exp(-{\min}^2(\frac{x}{N},1-\frac{x}{N})\cdot\frac{\bgamma}{2}))\cdot(1+\mathcal{O}(\frac{1}{\sqrt{\bgamma}})),\\
\pi_i &= (\frac{1}{N}\pm \mathcal{O}(g(\bgamma)))\cdot(1+\mathcal{O}(\frac{1}{\sqrt{\bgamma}})), \end{align*}
where the vanishing rate of $g(\bgamma)$ depends on $i$ as stated by the Lemma.
The high SNR behavior of many expressions below
 could be determined via the application of the results listed in Lemma~\ref{lem_Q_minus_Q}.

To complete our analysis of the QKD channel between Alice and Bob, it is necessary to find the likelihoods $p_{Y|\hX}(y|\hx)$ and the transition probabilities $p_{\hY|\hX}(\hy|\hx)$ for the soft-output and the hard-output mathematical models illustrated in Figure~\ref{fig_inf_theo_channels}. We proceed in a similar manner as from (\ref{equ_condXt_U}) to (\ref{equ_pdfX}), by first integrating over $U$, then truncating over the interval $[0,N[$. 

%%%%%%%%%%%%%%%%% A first lemma after Lemma 1 :-) %%%%%%%%%%%%
\begin{lemma}
\label{lem_U_tX}
Given Alice's frame is valid, i.e. $\tX \in [0,N[$, the density of $U$ becomes
\begin{equation} 
\label{equ_U_tX}
p_{U|\tX \in [0,N[}(u) ~=~ \frac{Q\left(\frac{-u}{\sigma}\right)-Q\left(\frac{N-u}{\sigma}\right)}
{\int_0^N \left[ Q\left(\frac{-t}{\sigma}\right)-Q\left(\frac{N-t}{\sigma}\right) \right]~dt}
=p_{U|\tY \in [0,N[}(u),~~~u \in [0,N[,
\end{equation}
where $p_{U|\tY \in [0,N[}(u)$ is the density of $U$ given that Bob's frame is valid.
Furthermore, the {\em a priori} probabilities $\{ \hpi_i\}_{i=0}^{N-1}$ when both frames are valid are given by
\begin{equation}
\label{equ_hpi}
\hpi_i = \pr(\hX=i|\tY \in [0,N)) 
= \frac{\int_0^N \left[ Q\left(\frac{i-u}{\sigma}\right)-Q\left(\frac{i+1-u}{\sigma}\right) \right]\cdot
\left[ Q\left(\frac{-u}{\sigma}\right)-Q\left(\frac{N-u}{\sigma}\right) \right]\, du}
{\int_0^N \left[ Q\left(\frac{-u}{\sigma}\right)-Q\left(\frac{N-u}{\sigma}\right) \right]^2\, du.} 
\end{equation}
\end{lemma}
\begin{proof}
Let us apply Bayes' rule, while dropping the subscripts to simplify the notation:
\[
p(u|\tX \in [0,N[) = \frac{\pr(\tX \in [0,N)|u) \times p_U(u)}{\pr(\tX \in [0,N[)}. 
\]
From (\ref{equ_condXt_U}), we get $\pr(\tX \in [0,N[|u)=Q\left(\frac{-u}{\sigma}\right)-Q\left(\frac{N-u}{\sigma}\right)$.\\ 
From (\ref{equ_pdfXt}), we get $\pr(\tX \in [0,N[=\frac{1}{N}\int_0^N \left[ Q\left(\frac{-t}{\sigma}\right)-Q\left(\frac{N-t}{\sigma}\right) \right]~dt.$\\
Finally, plugging $p_U(u)=1/N$ leads to the result announced by the lemma in~(\ref{equ_U_tX}). The equality  
$p_{U|\tX \in [0,N[}(u)=p_{U|\tY \in [0,N[}(u)$ is the result of the symmetry between Alice and Bob in our model.\\

The {\em a priori} probability $\hpi_i$ is derived after establishing the density of $\tX$ conditioned on a valid frame for Bob, $\tY \in [0,N)$.
\begin{align}
p_{\tX|\tY \in [0,N)}(\tx) &= \int_0^N p_{\tX|U,\tY \in [0,N[}(\tx|u) \cdot p_{U|\tY \in [0,N[}(u) ~du \nonumber \\ 
&= \int_0^N p_{\tX|U}(\tx|u) \cdot p_{U|\tY \in [0,N[}(u) ~du, \label{equ_pdf_tX_cond_tY}
\end{align}
where the two factors are given by (\ref{equ_condXt_U}) and (\ref{equ_U_tX}) respectively. 
The {\em a priori} probability $\hpi_i=\pr(X \in [i,i+1)|\tY \in [0,N))$ becomes, for $i=0, \dots,N-1$, 
\begin{align*}
\hpi_i &= \int_i^{i+1} p_{X|\tY \in [0,N[}(x)~dx 
= \int_i^{i+1} \frac{p_{\tX|\tY \in [0,N[}(x)}{\int_0^N p_{\tX|\tY \in [0,N)}(\tx)~d\tx}~dx \\
&= \frac{\int_{x=i}^{i+1} \int_{u=0}^N p_{\tX|U}(x|u) \cdot p_{U|\tY \in [0,N[}(u) ~du ~dx}
{\int_{\tx=0}^N \int_{u=0}^N p_{\tX|U}(\tx|u) \cdot p_{U|\tY \in [0,N[}(u) ~du ~d\tx} \\
&= \frac{\int_{u=0}^N \left[ Q\left(\frac{i-u}{\sigma}\right)-Q\left(\frac{i+1-u}{\sigma}\right) \right] 
\cdot p_{U|\tY \in [0,N)}(u) ~du ~dx}
{\int_{u=0}^N \left[ Q\left(\frac{-u}{\sigma}\right)-Q\left(\frac{N-u}{\sigma}\right) \right] \cdot p_{U|\tY \in [0,N)}(u) ~du ~d\tx}. 
\end{align*}
We obtain (\ref{equ_hpi}) after replacing $p_{U|\tY \in [0,N[}(u)$ by its expression from (\ref{equ_U_tX}).
\end{proof}

Lemma~\ref{lem_U_tX} tells that invalidating the cases where the photon falls outside the frame converts the uniform density $p_U(u)=\frac{1}{N}$ 
into a non-uniform density in (\ref{equ_U_tX}). 
Furthermore, the {\em a priori} probability $\pi_i$ of (\ref{equ_pi_i}) becomes $\hpi_i$ of (\ref{equ_hpi}) when adding
the condition that Bob's frame is valid. $\pi_i$ and $\hpi_i$ already take into account
that Alice has a valid frame. 
The next lemma leads to establishing the channel likelihood expression. 
%%%%%%%%%%%%%%%%% A second lemma %%%%%%%%%%%%
\begin{lemma}
\label{lem_U_cond_hX}
The conditional density of $U$ given $\hX=i$ is
\begin{equation}
\label{equ_u_hX}
p(u|\hX=i)= \frac{Q\left(\frac{i-u}{\sigma}\right)-Q\left(\frac{i+1-u}{\sigma}\right)}
{\int_0^{N} \left[ Q\left(\frac{i-t}{\sigma}\right)-Q\left(\frac{i+1-t}{\sigma}\right) \right]~dt},~~~u \in [0,N),     
\end{equation}
for $i=0, \dots, N-1$. Furthermore, when Bob gets a valid frame, the density of $U$ conditioned
on Alice's bin number $i$ is
\begin{equation}
\label{equ_u_hX_tY}
p(u|\hX=i,\tY \in [0,N))=
\frac{\left[ Q\left(\frac{i-u}{\sigma}\right)-Q\left(\frac{i+1-u}{\sigma}\right)\right] \cdot
\left[ Q\left(\frac{-u}{\sigma}\right)-Q\left(\frac{N-u}{\sigma}\right)\right]}
{\int_0^N \left[ Q\left(\frac{i-t}{\sigma}\right)-Q\left(\frac{i+1-t}{\sigma}\right)\right] \cdot
\left[ Q\left(\frac{-t}{\sigma}\right)-Q\left(\frac{N-t}{\sigma}\right)\right] ~dt}.
\end{equation}
\end{lemma}
\begin{proof}
The existence of $X$ and $\hX$, e.g. when writing $\hX=i$, requires that $\tX \in [0,N)$.
This hidden assumption should not be forgotten. 
By applying Bayes' rule, 
\[
p(u|\hX=i) = p(u|\hX=i, \tX \in [0,N))  = \frac{\pr(\hX=i|u, \tX \in [0,N)) \times p(u|\tX \in [0,N))}{\pi_i}. 
\]
The first term in the numerator can be developed as follows
\begin{align*}
\pr(\hX=i|u, \tX \in [0,N) = \pr(X \in [i,i+1)| u, \tX \in [0,N)) & = 
\frac{\pr(\tX \in [i,i+1)|u)}{\pr(\tX \in [0,N)|u)}\\  & = \frac{Q\left(\frac{i-u}{\sigma}\right)-Q\left(\frac{i+1-u}{\sigma}\right)}
{Q\left(\frac{-u}{\sigma}\right)-Q\left(\frac{N-u}{\sigma}\right)}. 
\end{align*}
The second term in the numerator is given in (\ref{equ_U_tX}) in Lemma~\ref{lem_U_tX}. 
After substituting the expression of $\pi_i$ from (\ref{equ_pi_i}), we get 
\begin{equation*}
p(u|\hX=i)= \frac{Q\left(\frac{i-u}{\sigma}\right)-Q\left(\frac{i+1-u}{\sigma}\right)}
{\int_i^{i+1} \left[ Q\left(\frac{-t}{\sigma}\right)-Q\left(\frac{N-t}{\sigma}\right) \right] \,dt},     
\end{equation*}
The reader is invited to prove via a change of variable that
\begin{equation}
\label{equ_exercise}
    \int_i^{i+1} \left[ Q\left(\frac{-t}{\sigma}\right)-Q\left(\frac{N-t}{\sigma}\right)\right] \,dt ~=~ 
    \int_0^N \left[ Q\left(\frac{i-t}{\sigma}\right)-Q\left(\frac{i+1-t}{\sigma}\right) \right] \,dt.  
\end{equation}
which leads to the result announced by the lemma in (\ref{equ_u_hX}).\\
%%%
The proof of (\ref{equ_u_hX_tY}) follows similar steps as for the proof of (\ref{equ_u_hX}).
Firstly, using Bayes' rule and (\ref{equ_tXtY_0N}) we get a conditional density of $U$, 
\begin{equation}
\label{equ_pdfU_valid}
p(u|\tX,\tY \in [0,N))= \frac{\left[Q\left(\frac{-u}{\sigma}\right)-Q\left(\frac{N-u}{\sigma}\right)\right]^2}
{\int_0^N \left[Q\left(\frac{-t}{\sigma}\right)-Q\left(\frac{N-t}{\sigma}\right)\right]^2 \,dt}. 
\end{equation}
Secondly, we solve the conditional probability of Alice's photon bins,
\begin{align*}
\pr(\hX=i|u, \tX,\tY \in [0,N)) &=  
\frac{\pr(\tX\in [i,i+1),\tY \in [0,N)|u)}{\pr(\tX,\tY \in [0,N)|u)} \\
&= \frac{\left[Q\left(\frac{i-u}{\sigma}\right)-Q\left(\frac{i+1-u}{\sigma}\right)\right] \cdot
\left[Q\left(\frac{-u}{\sigma}\right)-Q\left(\frac{N-u}{\sigma}\right)\right]}
{\left[Q\left(\frac{-u}{\sigma}\right)-Q\left(\frac{N-u}{\sigma}\right)\right]^2}. 
\end{align*}
Finally, we use the above expressions of $\pr(\hX=i|u, \tX,\tY \in [0,N))$ and $p(u|\hX,\tY \in [0,N))$,
and (\ref{equ_hpi}) from Lemma~\ref{lem_U_tX} in
\[
p(u|\hX=i, \tY \in [0,N))  = \frac{\pr(\hX=i|u, \tX,\tY \in [0,N)) \times p(u|\tX,\tY \in [0,n))}{\hpi_i}
\]
to reach (\ref{equ_u_hX_tY}) in this lemma.
\end{proof}
The existence of $Y$ assumes that $\tY \in [0,N)$, as we mentioned for $X$
in the proof of Lemma~\ref{lem_U_cond_hX}. We deliberately remind the reader of the condition
$\tY \in [0,N)$ in the subscript of the likelihood function in the next statement. 

%%%%%%%%%%%%%%%% Theorem establishing the likelihood expression %%%%%%%%%%
\begin{theorem}
\label{th_Y_cond_hX_tY}
Under the assumption that both Alice and Bob got valid frames, 
the soft-output QKD channel model likelihoods, $p(y|\hx)=p_{Y|\hX,\tY \in [0,N)}(y|\hx)$, 
have the following expression
\begin{equation}
\label{equ_Y_cond_hX_tY}
    p_{Y|\hX,\tY \in [0,N)}(y|\hx=i) ~=~\frac
    {\int_0^N \frac{1}{\sqrt{2\pi\sigma^2}} \exp\Bigl(-\frac{(y-u)^2}{2\sigma^2}\Bigr)\cdot \left[Q\left(\frac{i-u}{\sigma}\right)-Q\left(\frac{i+1-u}{\sigma}\right)\right] \,du}
    {\int_0^N \left[Q\left(\frac{-t}{\sigma}\right)-Q\left(\frac{N-t}{\sigma}\right)\right] \cdot \left[Q\left(\frac{i-t}{\sigma}\right)-Q\left(\frac{i+1-t}{\sigma}\right)\right] \,dt}, 
\end{equation}
for $i=0,\dots, N-1$, $y \in [0,N)$. For simplicity,
the likelihood in (\ref{equ_Y_cond_hX_tY}) will be denoted by $p(y|\hX=i)$ in next sections. 
\end{theorem}
\begin{proof}
We drop the subscripts in the density functions, when possible, to simplify the notations.
We start by a marginalization before truncating $p(\ty|u)$. 
\begin{align}
p(y|\hX=i, \tY \in [0,N)) &= \int_0^N p(y, u | \hX=i, \tY \in [0,N)) \,du \nonumber \\
&= \int_0^N p(y| u, \hX=i, \tY \in [0,N)) \cdot p(u|\hX=i, \tY \in [0,N)) \,du \nonumber \\
&= \int_0^N p(y| u) \cdot p(u|\hX=i, \tY \in [0,N)) \,du, \label{equ_Y_cond_hX_tY_temp}
\end{align}

The left factor $p(y | u)$ inside the integral in~(\ref{equ_Y_cond_hX_tY_temp}) is given by the truncation of the density in~(\ref{equ_condXt_U}) (replace $x$ by $y$) and the right factor was solved by Lemma~(\ref{lem_U_cond_hX}). 
\begin{align*}
p(y|\hX=i, \tY \in [0,N)) &= \int_0^N \frac{p(\ty=y|u)}{\int_0^N p(\ty|u) d\ty} \cdot p(u|\hX=i, \tY \in [0,N)) \,du,
\end{align*}
\begin{equation*}
= \int_0^N 
\frac{\frac{1}{\sqrt{2\pi\sigma^2}} \exp\Bigl(-\frac{(y-u)^2}{2\sigma^2}\Bigr)}{Q\left(\frac{-u}{\sigma}\right)-Q\left(\frac{N-u}{\sigma}\right)} 
\cdot 
\frac{\left[Q\left(\frac{i-u}{\sigma}\right)-Q\left(\frac{i+1-u}{\sigma}\right)\right] \cdot
\left[Q\left(\frac{-u}{\sigma}\right)-Q\left(\frac{N-u}{\sigma}\right)\right]}
{\int_0^N \left[Q\left(\frac{i-t}{\sigma}\right)-Q\left(\frac{i+1-t}{\sigma}\right)\right] \cdot
\left[Q\left(\frac{-t}{\sigma}\right)-Q\left(\frac{N-t}{\sigma}\right)\right] \,dt} \,du.
\end{equation*}
After simplifying the term $Q\left(\frac{-u}{\sigma}\right)-Q\left(\frac{N-u}{\sigma}\right)$ we reach the announced result. 
\end{proof}
%%%%%%%%%%%%%%% Corollary for Transition Probs %%%%%%%%%%%%
The transition probabilities $p_{i,j}=\pr(\hY=j|\hX=i)$ of the hard-output QKD channel model
are directly derived by integrating the conditional density function of the soft output $Y$ established by the previous theorem. 
\begin{corollary}
\label{cor_pij}
The probability that Bob's photon falls in bin $j$ given that Alice's photon fell in 
bin $i$ is given by
\begin{equation}
\label{equ_pij}
    p_{ij} ~=~\pr(\hY=j|\hX=i) = \frac
    {\int_0^N \left[Q\left(\frac{j-u}{\sigma}\right)-Q\left(\frac{j+1-u}{\sigma}\right)\right] \cdot \left[Q\left(\frac{i-u}{\sigma}\right)-Q\left(\frac{i+1-u}{\sigma}\right)\right] \,du}
    {\int_0^N \left[Q\left(\frac{-t}{\sigma}\right)-Q\left(\frac{N-t}{\sigma}\right)\right] \cdot \left[Q\left(\frac{i-t}{\sigma}\right)-Q\left(\frac{i+1-t}{\sigma}\right)\right] \,dt}, ~~ i, j \in \Z_N.
\end{equation}
\end{corollary}
\begin{proof}
Integrate (\ref{equ_Y_cond_hX_tY}) over Bob's photon position $y$ from $j$ to $j+1$,
then switch the two integrals to get the result announced by this corollary. 
\end{proof}

We complete this section by establishing the expression of the {\em a posteriori} probability
useful for soft-decision decoding, e.g., for belief-propagation decoding of low-density 
parity-check codes, for ordered-statistics decoding of linear block codes, or 
Viterbi decoding of convolutional codes. 
Let $APP(i)=APP(\hX=i)=\pr(\hX=i|Y=y)$ be the {\em a posteriori} probability of
Alice's photon bin number $i$, for $i=0 \dots N-1$. The next theorem gives
the expression $APP(i)$, which is used in our proposed coding/decoding schemes
in Section~\ref{sec_codes}. 

%%%%%%%%% The APP Theorem %%%%%%%%%%
\begin{theorem}
\label{th_app}
Given the photon position $Y=y$ on Bob's side, the probability for Alice's photon to belong
to bin number $i$ is
\begin{equation}
\label{equ_app}
APP(i) =  \frac{
\int_0^N \frac{1}{\sqrt{2\pi\sigma^2}} e^{-\frac{(y-u)^2}{2\sigma^2}}\cdot \left[Q\left(\frac{i-u}{\sigma}\right)-Q\left(\frac{i+1-u}{\sigma}\right)\right] \,du}
{\int_0^N \frac{1}{\sqrt{2\pi\sigma^2}} e^{-\frac{(y-t)^2}{2\sigma^2}}\cdot \left[Q\left(\frac{-t}{\sigma}\right)-Q\left(\frac{N-t}{\sigma}\right)\right] \,dt},
~~ i \in \Z_N.
\end{equation}
\end{theorem}
\begin{proof}
Keeping in mind that $\tX,\tY \in [0,N)$, 
apply Bayes' rule to get
\[
\pr(\hX=i|Y=y) = \frac{p(y|\hX=i) \cdot \pr(\hX=i)}{p(y)}. 
\]
The result announced by the theorem is then found in three steps.\\ 
(i) Use (\ref{equ_Y_cond_hX_tY}) from Theorem~\ref{th_Y_cond_hX_tY} for $p(y|\hX=i)$.\\
(ii) Use (\ref{equ_hpi}) for the a priori $\pr(\hX=i)$.\\ 
(iii) Finally, $p(y)=p_{Y|\tX \in [0,N)}(y)=p_{\tY|\tX \in [0,N)}(\ty=y)/\int_0^N p_{\tY|\tX \in [0,N)}(t)\,dt$ 
after truncating the density of $\tY$.
The density $p_{\tY|\tX \in [0,N)}(\ty)$ is found in (\ref{equ_pdf_tX_cond_tY}) while
switching the letters $x$ (resp. $X$) and $y$ (resp. $Y$), in conjunction with (\ref{equ_condXt_U}) and (\ref{equ_U_tX}),
\begin{equation}
\label{equ_pdfY_valid}  %% both Alice's and Bob's frames are valid
p(y)=  
\frac{
\int_0^N \frac{1}{\sqrt{2\pi\sigma^2}} e^{-\frac{(y-u)^2}{2\sigma^2}}\cdot \left[Q\left(\frac{-u}{\sigma}\right)-Q\left(\frac{N-u}{\sigma}\right)\right] \,du}
{\int_0^N \left[Q\left(\frac{-t}{\sigma}\right)-Q\left(\frac{N-t}{\sigma}\right)\right]^2 \,dt}.
\end{equation}
\end{proof}

For consistency, the reader could check that $p(y)$ given at the end of the proof of Theorem~\ref{th_app} is also equal to $\sum_{i=0}^{N-1} \hat{\pi}_i \cdot p(y|\hX=i)$ from (\ref{equ_hpi}) and (\ref{equ_Y_cond_hX_tY}). Figures~\ref{fig_likelihoods10dB} and~\ref{fig_likelihoods25dB} plot the likelihoods
$p(y|\hX=i)$ at low SNR $\gamma=10$~dB (low photon detector precision) and a relatively higher SNR $\gamma=25$~dB (higher photon detector precision), respectively. At low SNR, $p(y|\hX=i)$ has a Gaussian shape. The shape
tends to become square at high SNR. The {\em a posteriori} probabilities
$APP(i)$, $i \in \Z_N$, have a plot similar to the channel likelihoods. 

\begin{figure}[hbt]
   \centering
   \includegraphics[width=7cm,angle=270]{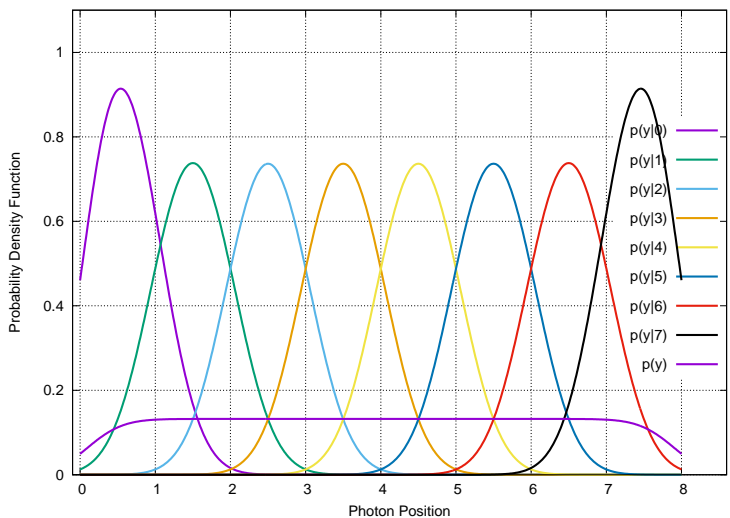}
   \caption{Soft-Output channel likelihoods, N=8 bins per frame, SNR=10~dB.}
\label{fig_likelihoods10dB}
\end{figure}

\begin{figure}[hbt]
   \centering
   \includegraphics[width=7cm,angle=270]{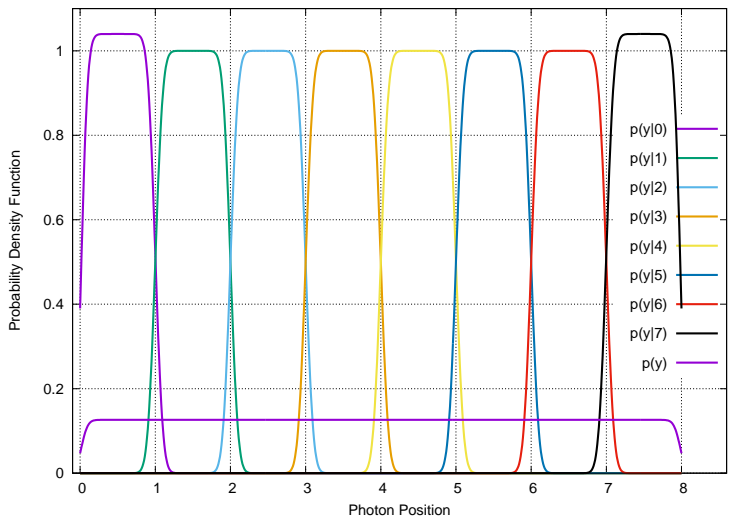}
   \caption{Soft-Output channel likelihoods, N=8 bins per frame, SNR=25~dB.}
\label{fig_likelihoods25dB}
\end{figure}
%%------------------------------------------------------------------
\section{Secret Key Information Rates \label{sec_mut_info}}
\subsection{Mutual Information between Raw Keys}
Firstly, we consider the mutual information of the algebraic hard-output channel defined by the transition probability $p(\hy|\hx)=\pr(\hY=\hy|\hX=\hx)$. In Corollary~\ref{cor_pij}, the expression
of $p_{ij}=p(\hY=j|\hX=i)$ was established. Hence, we can directly
compute the mutual information as follows:
\begin{align}
I(\hX;\hY) &=H(\hY)-H(\hY|\hX) \nonumber \\ 
&=-\sum_{i=0}^{N-1} \pr(\hY=i) \log(\pr(\hY=i)) \nonumber \\
&~~~+ \sum_{i=0}^{N-1} \pr(\hX=i) 
\sum_{j=0}^{N-1} \pr(\hY=j|\hX=i) \log(\pr(\hY=j|\hX=i)) \nonumber \\ 
&=-\sum_{i=0}^{N-1} \hpi_i \log(\hpi_i) 
~+ \sum_{i=0}^{N-1} \hpi_i \sum_{j=0}^{N-1} p_{ij} \log(p_{ij}),  \label{equ_I_hX_hY}
\end{align}
where {\em a priori} $\hpi_i$ of $\hX$ and $\hY$ is found in (\ref{equ_hpi}). The plot of $I(\hX;\hY)$ expressed in bits versus the signal-to-noise ratio is depicted in Figure~\ref{fig_mutinfo_discrete}, for 2, 3, and 4 coded bits per transmitted photon. 
As expected, the curves go towards the asymptote $H(\hX)$ at high signal-to-noise ratio. In fact, the entropy $-\sum_{i=0}^{N-1} \hpi_i \log(\hpi_i)$ is very stable even at low SNR and could be well approximated by $\log(N)$. The summation in (\ref{equ_I_hX_hY}) could be truncated to neighboring bins or to bins within an integer distance less than $D$,
\begin{equation}
\label{equ_approxI_discrete}
I(\hX;\hY) \approx  \log(N) + \frac{1}{N} \sum_{|i-j|\le D} p_{ij} \log(p_{ij}).
\end{equation}
The simplification (\ref{equ_approxI_discrete}) is an excellent approximation down to $\gamma \ge 10$~dB for $D=1$ only and it extends to $\gamma \ge 5$~dB for $D=2$. The next proposition gives more insight into the behavior of the {\em a priori} and the transition probabilities, and the discrete channel mutual information in the low-noise regime. 

\begin{figure}[h]
   \centering
   \includegraphics[width=9cm,angle=270]{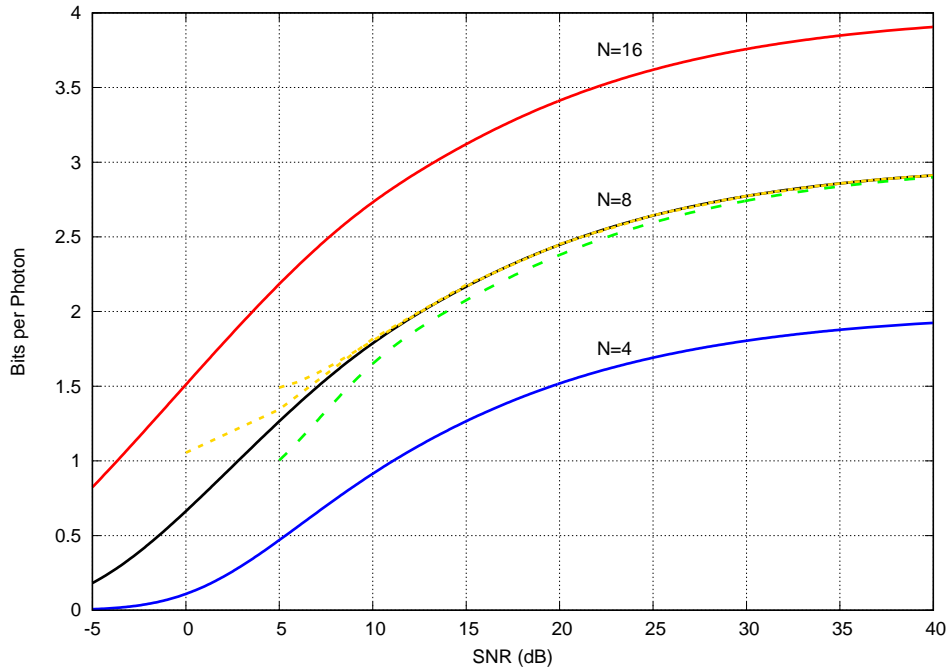}
   \caption{Mutual information $I(\hX;\hY)$ of the algebraic hard-output TE-QKD channel, for $N=4, 8, 16$ bins per frame.}
\label{fig_mutinfo_discrete}
\end{figure}

\begin{proposition}
\label{prop_high_snr}
At high signal-to-noise ratio, when $\sigma^2 \ll 1$, we have the following results:\\
a) The transition probability of the hard-output QKD channel established in Corollary~\ref{cor_pij} satisfies:\\
At the two extremal bins, $i=0$ and $i=N-1$, we have 
\begin{equation}
\label{equ_p01_small_sigma}
p_{0,1}=\frac{\frac{\sigma}{\sqrt{\pi}}+\mathcal{O}(e^{-\frac{\gamma}{4}})}
{1-\frac{1+\sqrt{2}}{2\sqrt{\pi}}\cdot\sigma+\mathcal{O}(e^{-\frac{\gamma}{4}})},
\end{equation}
where $p_{0,1}=p_{N-1,N-2}=1-p_{0,0}=1-p_{N-1,N-1}$, and $\sigma=1/\sqrt{\gamma}$.\\ 
At the middle bins, $i=2 \ldots N-2$, we have
\begin{equation}
\label{equ_p12_small_sigma}
p_{1,2}=\frac{\frac{\sigma}{\sqrt{\pi}}+\mathcal{O}(e^{-\frac{\gamma}{4}})}
{1-\mathcal{O}(e^{-\frac{\gamma}{4}})},
\end{equation}
where $p_{1,2}=p_{i,i+1}=p_{i,i-1}=(1-p_{i,i})/2$.
All other transition probabilities $p_{i,j}$ for $|i-j|\ge 2$ are
$\mathcal{O}(e^{-\frac{\gamma}{4}})$ and can be forced to $0$ in any numerical calculation at high SNR.\\
b) The {\em a priori} probabilities established in Lemma~\ref{lem_U_tX} satisfy\\
At the two extremal bins, $i=0$ and $i=N-1$, we have 
\begin{equation}
\label{equ_hpi0_small_sigma}
\hpi_0=\hpi_{N-1}=\frac{1-\frac{1+\sqrt{2}}{2\sqrt{\pi}}\cdot\sigma+\mathcal{O}(e^{-\frac{\gamma}{4}})}
{N\cdot \left[ 1-\frac{1+\sqrt{2}}{\sqrt{\pi}}\cdot\bsigma+\mathcal{O}(e^{-\frac{\gamma}{4}})\right]},
\end{equation}
where the numerator includes $\sigma$ but the denominator involves
$\bsigma=\sigma/N$. For the middle bins, with $i=2 \ldots N-2$, we have
\begin{equation}
\label{equ_hpii_small_sigma}
\hpi_i=\frac{1-\mathcal{O}(e^{-\frac{\gamma}{4}})}
{N\cdot \left[ 1-\frac{1+\sqrt{2}}{\sqrt{\pi}}\cdot\bsigma+\mathcal{O}(e^{-\frac{\gamma}{4}})\right]}. 
\end{equation}
c) Following a) and b), the mutual information of the discrete-input discrete-output
QKD channel given by (\ref{equ_I_hX_hY}) becomes
\begin{align}
I(\hX;\hY) &=\frac{N-2\beta\sigma}{N(1-2\beta\bsigma)} \log\left[N(1-2\beta\bsigma)\right]
-\frac{2(1-\beta\sigma)}{N(1-2\beta\bsigma)} \log(1-\beta\sigma) \nonumber\\
&-\frac{2(1-\beta\sigma)}{N(1-2\beta\bsigma)} H_2\left(\frac{\sigma/\sqrt{\pi}}{1-\beta\sigma}\right)
-\frac{(N-2)}{N(1-2\beta\bsigma)}H_3\left(\frac{\sigma}{\sqrt{\pi}}\right)+ \mathcal{O}(e^{-\frac{\gamma}{4}}), \label{equ_approx_IhXhY_H2_H3}
\end{align}
where $\beta=\frac{1+\sqrt{2}}{2\sqrt{\pi}}$, $\sigma=1/\sqrt{\gamma}$, and $\bsigma=\sigma/N$. 
\end{proposition}
\begin{proof}
For the sake of space, we only show the detailed proof for the denominator of $p_{0,1}$ (also equal to the numerator of $\hpi_0$). All other results are found using similar calculus techniques. The denominator of $p_{0,1}$ from Corollary~\ref{cor_pij} ($i=0, j=1$) is equal to the integral
\begin{align*}
I_5 &=\int_0^N \left[Q\left(\tfrac{-t}{\sigma}\right)-Q\left(\tfrac{N-t}{\sigma}\right)\right]\cdot \left[Q\left(\tfrac{-t}{\sigma}\right)-Q\left(\tfrac{1-t}{\sigma}\right)\right] ~dt\\
&= \int_0^N \left[1-Q\left(\tfrac{t}{\sigma}\right)-Q\left(\tfrac{N-t}{\sigma}\right)\right]\cdot \left[Q\left(\tfrac{t-1}{\sigma}\right)-Q\left(\tfrac{t}{\sigma}\right)\right] ~dt\\
&= \int_0^N Q\left(\tfrac{t-1}{\sigma}\right) ~dt - \int_0^N Q\left(\tfrac{t}{\sigma}\right) ~dt + \int_0^N Q^2\left(\tfrac{t}{\sigma}\right) ~dt
- \int_0^N Q\left(\tfrac{t}{\sigma}\right) Q\left(\tfrac{t-1}{\sigma}\right) ~dt\\
&- \int_0^N Q\left(\tfrac{N-t}{\sigma}\right) Q\left(\tfrac{t-1}{\sigma}\right) ~dt
+ \int_0^N Q\left(\tfrac{N-t}{\sigma}\right) Q\left(\tfrac{t}{\sigma}\right) ~dt\\
&=\underbrace{\int_{-1}^0 Q\left(\tfrac{t}{\sigma}\right) ~dt}_{(i)}  
- \underbrace{\int_{N-1}^N Q\left(\tfrac{t}{\sigma}\right) ~dt}_{(ii)} 
+ \underbrace{\int_0^N Q^2\left(\tfrac{t}{\sigma}\right) ~dt}_{(iii)}
- \underbrace{\int_0^1 Q\left(\tfrac{t}{\sigma}\right) (1-Q\left(\tfrac{1-t}{\sigma}\right)) ~dt}_{(iv)} \\
&-\underbrace{\int_1^N Q\left(\tfrac{t}{\sigma}\right) Q\left(\tfrac{t-1}{\sigma}\right) ~dt}_{(v)}
-\underbrace{\int_0^1 Q\left(\tfrac{N-t}{\sigma}\right) (1-Q\left(\tfrac{1-t}{\sigma}\right)) ~dt}_{(vi)}\\
&-\underbrace{\int_1^N Q\left(\tfrac{N-t}{\sigma}\right) Q\left(\tfrac{t-1}{\sigma}\right) ~dt}_{(vii)} 
+ \underbrace{\int_0^N Q\left(\tfrac{N-t}{\sigma}\right) Q\left(\tfrac{t}{\sigma}\right) ~dt}_{(viii)}.
\end{align*}
Now we solve the elementary integrals (i)-(viii) one by one.\\
Using (\ref{equ_antider}),
$(i)=0-\tfrac{\sigma}{\sqrt{2\pi}}+Q(\tfrac{-1}{\sigma})+\tfrac{\sigma}{\sqrt{2\pi}}e^{-\tfrac{\gamma}{2}}=1-\tfrac{\sigma}{\sqrt{2\pi}}+\mathcal{O}(e^{-\tfrac{\gamma}{2}})$.
Using the fact that $Q(x)$ is a monotone decreasing function, 
then $(ii)=\mathcal{O}(e^{-(N-1)^2\tfrac{\gamma}{2}})$. The third integral is directly solved via (\ref{equ_antider2}): 
$(iii)=NQ^2(\tfrac{N}{\sigma})-\sigma\sqrt{\tfrac{2}{\pi}}Q(\tfrac{N}{\sigma})e^{-N^2\tfrac{\gamma}{2}}+\tfrac{\sigma}{\sqrt{\pi}}Q(\tfrac{N\sqrt{2}}{\sigma})-0+\sigma\sqrt{\tfrac{2}{\pi}}\cdot \tfrac{1}{2}-\tfrac{\sigma}{\sqrt{\pi}}\cdot \tfrac{1}{2}$. Then we find 
$(iii)=\left(\tfrac{\sqrt{2}-1}{2\sqrt{\pi}}\right) \sigma + \mathcal{O}(e^{-N^2\tfrac{\gamma}{2}})$.
$(iv)=I_1-I_3=\tfrac{\sigma}{\sqrt{2\pi}}+\mathcal{O}(e^{-\tfrac{\gamma}{4}})$.
For (v), $t^2+(t-1)^2 \ge 1$ in the interval $[1,N]$, then we have
$(v)=\mathcal{O}(e^{-\tfrac{\gamma}{2}})$. The first part of (vi) is
$\mathcal{O}(e^{-(N-1)^2\tfrac{\gamma}{2}})$ and the second part is
also $\mathcal{O}(e^{-(N-1)^2\tfrac{\gamma}{2}})$ because $(N-t)^2+(1-t)^2 \ge (N-1)^2$ in the interval $[0,1]$. So $(vi)=\mathcal{O}(e^{-(N-1)^2\tfrac{\gamma}{2}})$. Applying similar arguments, we get $(vii)=\mathcal{O}(e^{-(N-1)^2\tfrac{\gamma}{4}})$ and $(viii)=\mathcal{O}(e^{-N^2\tfrac{\gamma}{4}})$.
Combining (i)-(viii) yields $I_5=1-\tfrac{1+\sqrt{2}}{2\sqrt{\pi}}\sigma+\mathcal{O}(e^{-\tfrac{\gamma}{4}})$ as announced. 
\end{proof}

At high signal-to-noise ratio, Proposition~\ref{prop_high_snr}-a) shows how fast $p_{i,j}$ converges to $\tfrac{\sigma}{\sqrt{\pi}}$. The latter is a one-sided probability of error and it is half the double-sided probability of error stated in Proposition~\ref{prop_pe_uncoded}. As expected, $\hpi_i$ converges to $1/N$ much faster for inner bins as found in Proposition~\ref{prop_high_snr}-b). The high-SNR expression of $I(\hX;\hY)$
established in Proposition~\ref{prop_high_snr}-c) perfectly fits the exact mutual information of the discrete channel down to $\gamma=10$ dB and then diverges at low SNR below 10 dB. The binary entropy function represents the extremal bins error. The ternary entropy function carries the inner bins error. Expression~(\ref{equ_approx_IhXhY_H2_H3}) is a quick method to evaluate $I(\hX;\hY)$ at moderate and high signal-to-noise ratios without performing any integration.\\

One could ask how good is the approximated mutual information if the TE-QKD discrete channel
is assumed to have a circular transition probability matrix. Under the assumptions of Proposition~\ref{prop_pe_uncoded}, we take: 
1- $X$ and $Y$ are Gaussian, 2- all bins are equiprobable, and 3- the error probability of the discrete-input discrete-output channel is dominated by events where $X$ and $Y$ are separated by one or two bins only. According to Theorem~7.2.1 in~\cite[Ch.~7.2]{books:CT2006}, the expression for a circular discrete channel is
\begin{equation}
\label{equ_approx_discrete_circular}
I(\hX;\hY) \approx \log(N)+\sum_{j=-2}^2 p_{0j} \log_2(p_{0j}),    
\end{equation}
where $p_{01}=p_{0,-1} \approx \int_0^1 2p_1p_2~dv=2(I_1-I_2-I_3)$, $p_{02}=p_{0,-2} \approx \int_0^1 2p_2p_3~dv=2I_3$, 
and $p_{00}=1-2p_{01}-2p_{02}$. All three high-SNR approximations (\ref{equ_approx_IhXhY_H2_H3}), (\ref{equ_approxI_discrete}) with $D=2$, and (\ref{equ_approx_discrete_circular}) are respectively shown in dotted lines from top to bottom on Figure~\ref{fig_mutinfo_discrete} for $N=8$ bins per frame. (\ref{equ_approx_IhXhY_H2_H3}) and (\ref{equ_approxI_discrete}) follows the exact mutual information $I(\hX;\hY)$ at high SNR. (\ref{equ_approx_discrete_circular}) is not tight enough at $N=8$ but becomes tighter for $N \ge 16$ bins per frame. \\ 

The second step in this section is to compute the mutual information for the soft-output TE-QKD channel.
We chose to write $I(\hX;Y)=H(\hX)-H(\hX|Y)$. The second expression after flipping $X$ and $Y$
based on differential entropy is also equivalent from numerical stability point of view
and has all its terms established in the previous section. We prefer the mutual information where the high-SNR asymptote is visible, hence
\begin{equation}
\label{equ_I_hX_Y}
I(\hX;Y)=H(\hX)-H(\hX|Y) =  H(\hX) 
+ \sum_{i=0}^{N-1} \hpi_i \int_0^N p(y|\hx=i) \log_2(APP(i)) ~dy, 
\end{equation}
where the {\em a priori} $\hpi_i$ is from (\ref{equ_hpi}), the likelihood $p(y|\hx=i)$ is from (\ref{equ_Y_cond_hX_tY}), and the {\em a posteriori} $APP(i)$ is from (\ref{equ_app}). Figure~\ref{fig_mutinfo_soft} shows the mutual information $I(\hX;Y)$ versus normalized SNR $\bgamma$ for different
number of bins per frame. The red upper envelope is established by 
Theorem~\ref{th_capa} in the next section. It corresponds to the maximal mutual information achievable on the TE-QKD channel. 

\begin{figure}[h]
   \centering
   \includegraphics[width=10cm,angle=270]{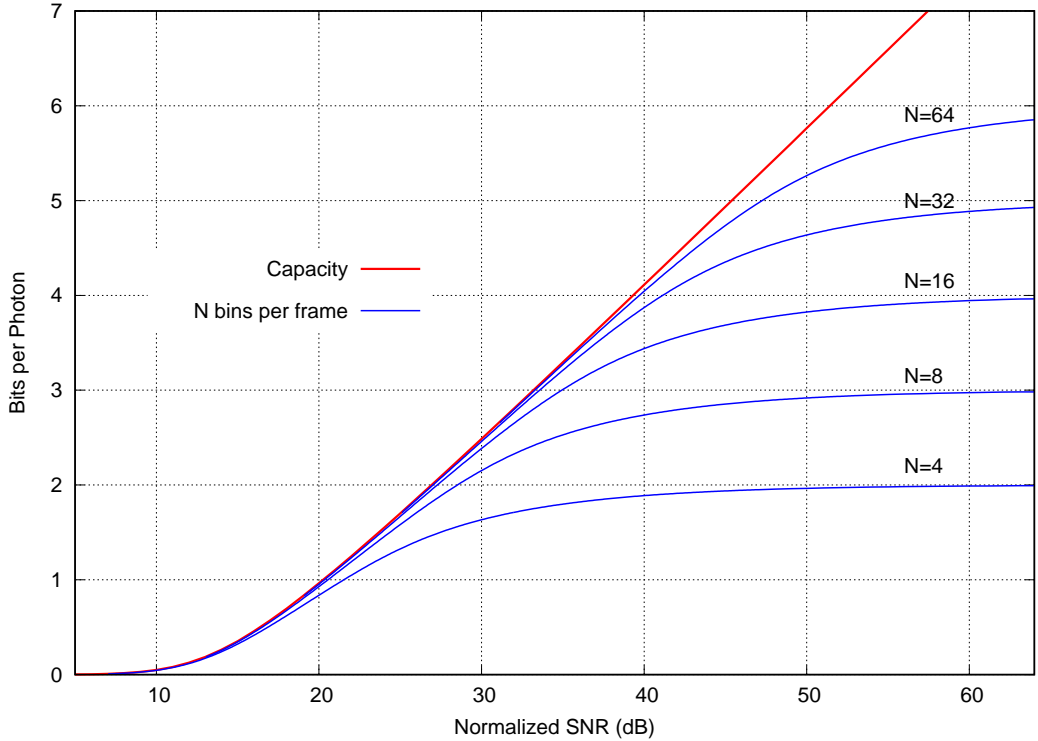}
   \caption{Mutual information $I(\hX;Y)$ of the soft-output QKD channel, for $N=4, 8, 16$ bins per frame.}
\label{fig_mutinfo_soft}
\end{figure}
%%------------------------------------------------------------------
\subsection{Maximal Secrecy Rate \label{sec_capacity}}
The random variables $X$, $\hX$, and $Y$ form a Markov chain
$X \rightarrow \hX \rightarrow Y$. Therefore, the data processing inequality \cite[Ch.~2]{books:CT2006} yields
\begin{equation}
I(\hX;Y)~\le~I(X;Y), ~~~\forall~N\ge 2.  
\end{equation}
Consequently, for any value of the number $N$ of bins per frame, the rate of our channel $p_{Y|\hX}(y|\hx)$ is always bounded from above by the rate of the continuous-input continuous-output channel $p_{Y|X}(y|x)$ corresponding to 
a continuum of zero-measure bins in Alice's frame. Thus, by determining the mutual information $I(X;Y)$
we get the maximal secrecy rate of the photon channel between Alice and Bob.\\

Without loss of generality, assume that the frame size is $1$, instead of $N$.
Now, the problem is to find $I(X;Y)$ where $X$ and $Y$ are truncated versions of the original photon positions, $X~=~\tX \in [0,1)$, $Y~=~\tY \in [0,1)$. 
The model in~(\ref{equ_XYmodel}) becomes $\tX=U+\mathcal{N}(0,\sigma^2)$,
$\tY=U+\mathcal{N}(0,\sigma^2)$, $U$ is uniform in $[0,1)$, and 
the two additive Gaussian noises are independent. The normalized signal-to-noise ratio is naturally defined by $\gamma=\bgamma=1/\sigma^2$ under this context of infinite number of bins and a frame of unit length.\\ 

\begin{theorem}
\label{th_capa}
The maximal secrecy rate $I(X;Y)$ is given by
\begin{align}
I(X;Y) &= h(Y)-h(Y|X) \\
       &=-\int_0^1 p(y)~\log(p(y))~dy ~+~
\int_0^1 p(x) \int_0^1 p(y|x)~\log(p(y|x))~dx dy \label{equ_IXY_0_1}
\end{align}
where $p(x)$ and $p(y)$ are from~(\ref{equ_pdfY_valid}) after
replacing the frame size $N$ by $1$, 
\begin{equation}
\label{equ_px_tX_tY} 
p(x)=p_{X|\tX,\tY \in [0,1)}(x)=  
\frac{
\int_0^1 \frac{1}{\sqrt{2\pi\sigma^2}} e^{-\frac{(x-u)^2}{2\sigma^2}} \left[Q\left(\frac{-u}{\sigma}\right)-Q\left(\frac{1-u}{\sigma}\right)\right] \,du}
{\int_0^1 \left[Q\left(\frac{-t}{\sigma}\right)-Q\left(\frac{1-t}{\sigma}\right)\right]^2 \,dt},
\end{equation}
\begin{equation}
\label{equ_py_tX_tY} 
p(y)=p_{Y|\tX,\tY \in [0,1)}(y)=  
\frac{
\int_0^1 \frac{1}{\sqrt{2\pi\sigma^2}} e^{-\frac{(y-u)^2}{2\sigma^2}} \left[Q\left(\frac{-u}{\sigma}\right)-Q\left(\frac{1-u}{\sigma}\right)\right] \,du}
{\int_0^1 \left[Q\left(\frac{-t}{\sigma}\right)-Q\left(\frac{1-t}{\sigma}\right)\right]^2 \,dt},  
\end{equation}
\begin{equation}
\label{equ_p_y_x_capa}
p(y|x)=\frac{1}{\sqrt{4\pi\sigma^2}} e^{-\frac{(y-x)^2}{4\sigma^2}} 
\frac{\left[Q\left(\frac{0-(x+y)/2}{\sigma/\sqrt{2}}\right)
-Q\left(\frac{1-(x+y)/2}{\sigma/\sqrt{2}}\right)\right]}
{\int_0^1 \frac{1}{\sqrt{2\pi\sigma^2}} e^{-\frac{(x-u)^2}{2\sigma^2}} \left[Q\left(\frac{-u}{\sigma}\right)-Q\left(\frac{1-u}{\sigma}\right)\right] \,du}. 
\end{equation}
\end{theorem}
\begin{proof}
We complete the proof by finding the expression of the conditional density $p(y|x)$.
Indeed, we can write after marginalizing and applying Bayes' rule
\begin{equation}
\label{equ_p_y_x_proof}
p(y|x)=p_{Y|X,\tX,\tY \in [0,1)}(y|x)= \int_0^1 p(y,u|x)~du = 
\int_0^1 \frac{p(u)p(x|u)p(y|u)}{p(x)} du.
\end{equation}
In the above integral expression we used the fact that $p(x,y|u)=p(x|u)p(y|u)$ as a result of the model defined by~(\ref{equ_XYmodel}). 
In~(\ref{equ_p_y_x_proof}), both $p(x|u)$ and $p(y|u)$ are from~(\ref{equ_pdfX_condU}),
$p(x)$ is from~(\ref{equ_pdfY_valid}), and finally $p(u)=p_{U|\tX,\tY\in [0,1)}(u)$ 
is found in~(\ref{equ_pdfU_valid}), all after substituting $1$ to $N$.
After simplifying the integrand of~(\ref{equ_p_y_x_proof}), we get $p(y|x)$ as stated by~(\ref{equ_p_y_x_capa}). 
\end{proof}

\begin{corollary}
\label{cor_log_formula}
At high signal-to-noise ratio, i.e. $\sigma^2 \ll 1$ or equivalently $\gamma=\frac{1}{\sigma^2} \gg 1$, the maximal secrecy rate satisfies 
\begin{align}
I(X;Y) &=(1+\mathcal{O}(\tfrac{1}{\sqrt{\gamma}})) \cdot \frac{1}{2}\log\left( \frac{\gamma}{4\pi e}\right) ~+~\mathcal{O}\left((\tfrac{1}{\sqrt{\gamma}})^{\alpha}\right)\sim \frac{1}{2}\log\left( \frac{\gamma}{4\pi e}\right),~~~\forall \alpha \in ]0,1[.
\label{equ_capa_approx}
\end{align}
\end{corollary}
\begin{proof} The proof is based on a Babylonian approach with heavy calculus.
Let us first give a sketch on how the limit is guessed. 
By applying Lemma~\ref{lem_Q_minus_Q} and some extra algebra, when $\gamma \gg 1$, 
we get that $p(x) \rightarrow 1$,
$p(y) \rightarrow 1$, 
and $p(y|x) \rightarrow \frac{1}{\sqrt{4\pi\sigma^2}} e^{-\frac{(y-x)^2}{4\sigma^2}}$, in (\ref{equ_px_tX_tY}), (\ref{equ_py_tX_tY}), and (\ref{equ_p_y_x_capa}) respectively,
for $x, y \in ]0,1[$. 
Then, the differential entropy $h(Y) \rightarrow 0$, 
$h(Y|X) \rightarrow h(\mathcal{N}(0,2\sigma^2))
=\frac{1}{2}\log(4\pi e \sigma^2)$, 
so the maximal secrecy rate satisfies $I(X;Y)=h(Y)-h(Y|X) \to \frac{1}{2}\log\left( \frac{\gamma}{4\pi e}\right)$. This ends a simple but a clear sketch on how all involved densities and the maximal mutual information are converging at high SNR.

The denominator of $p(x)$ and $p(y)$ is given in Lemma~\ref{lem_Q_minus_Q} at high SNR as $1-\mathcal{O}(\tfrac{1}{\sqrt{\gamma}})$. The numerator of $p(y|x)$ is $f_{\sigma/\sqrt{2}}((x+y)/2)=1-\mathcal{O}(\exp(-\min^2 \cdot\gamma))$ for $(x+y)/2 \in ]0,1[$, where $\min$ is $\min((x+y)/2,1-(x+y)/2)$ from Lemma~\ref{lem_Q_minus_Q}. Also, $f_{\sigma/\sqrt{2}}(0)=f_{\sigma/\sqrt{2}}(1)=\tfrac{1}{2}-\mathcal{O}(\exp(-\gamma))$. 
The last item to solve to get the result of this corollary is the integral in the numerator of $p(x)$, the numerator of $p(y)$, and the denominator of $p(y|x)$. Define the following integral 
\begin{equation}
I_6=I_6(x)=I_6(1-x)=\int_0^1 \frac{1}{\sqrt{2\pi\sigma^2}} ~e^{-\tfrac{(x-u)^2}{2\sigma^2}} \left[ Q\left(\tfrac{-u}{\sigma}\right) - Q\left( \tfrac{1-u}{\sigma}\right)\right] ~du, ~~~x \in [0,1]. 
\end{equation}
The integral $I_6$ has no closed-form expression. Firstly, we study $I_6(x)$ at $x=0$ (identical at $x=1$).
\begin{align}
I_6(0) &=\int_0^1 \frac{1}{\sqrt{2\pi\sigma^2}} ~e^{-\tfrac{u^2}{2\sigma^2}} \left[ 1-Q\left(\tfrac{u}{\sigma}\right) - Q\left( \tfrac{1-u}{\sigma}\right)\right] ~du \nonumber\\
&= \tfrac{1}{2}-Q\left(\tfrac{1}{\sigma}\right)-\int_0^1 \frac{1}{\sqrt{2\pi\sigma^2}} ~e^{-\tfrac{u^2}{2\sigma^2}} Q\left(\tfrac{u}{\sigma}\right)~du-\int_0^1 \frac{1}{\sqrt{2\pi\sigma^2}} ~e^{-\tfrac{u^2}{2\sigma^2}} Q\left(\tfrac{1-u}{\sigma}\right)~du \nonumber\\
&= \tfrac{1}{2}-\mathcal{O}(e^{-\tfrac{\gamma}{2}})-\tfrac{1}{8}
-\tfrac{1}{2} Q^2\left(\tfrac{1}{\sigma}\right)-\mathcal{O}(e^{-\tfrac{\gamma}{4}})~=~\tfrac{3}{8}-\mathcal{O}(e^{-\tfrac{\gamma}{4}}).
\end{align}
Secondly, we study $I_6(x)$ for $x \in ]0,1[$. We use calculus tools similar to those used in the proofs of Lemma~\ref{lem_Q_minus_Q} and Proposition~\ref{prop_high_snr} to obtain
\begin{align}
I_6(x) &=Q\left(\tfrac{-x}{\sigma}\right)-Q\left(\tfrac{1-x}{\sigma}\right)-\int_0^1 \frac{1}{\sqrt{2\pi\sigma^2}} ~e^{-\tfrac{(x-u)^2}{2\sigma^2}} Q\left(\tfrac{u}{\sigma}\right)~du-\int_0^1 \frac{1}{\sqrt{2\pi\sigma^2}} ~e^{-\tfrac{(x-u)^2}{2\sigma^2}} Q\left(\tfrac{1-u}{\sigma}\right)~du \nonumber\\
&= f_{\sigma}(x) - \mathcal{O}(e^{-\tfrac{x^2}{4\sigma^2}}) 
- \mathcal{O}(e^{-\tfrac{(1-x)^2}{4\sigma^2}})  
~=~1-\mathcal{O}(e^{-\min^2(x,1-x) \cdot \tfrac{\gamma}{4}}). \label{equ_I6_one}
\end{align}
Now we are ready to transform the expression of $I(X;Y)$ in the small-noise regime given that the behavior of all densities is solved:
\[
p(x)=\frac{I_6(x)}{1-\mathcal{O}(\tfrac{1}{\sqrt{\gamma}})},~~
p(y)=\frac{I_6(y)}{1-\mathcal{O}(\tfrac{1}{\sqrt{\gamma}})},~~
\text{and}~~p(y|x)=\frac{e^{-\tfrac{(y-x)^2}{4\sigma^2}}}{\sqrt{4\pi\sigma^2}} \cdot \frac{f_{\sigma/\sqrt{2}}((x+y)/2)}{I_6(x)}. 
\]
For $x,y \in [0,1]$, we distinguish between the behavior of $I_6(x)$, $I_6(y)$,
and $f_{\sigma/\sqrt{2}}((x+y)/2)$ near the extremal points $0$ and $1$ and inside the interval. Hence, we decompose the interval as $[0,1]= [0,\delta] \cup [\delta, 1-\delta] \cup [1-\delta, 1]$ for integration. The parameter $\delta$ should vanish at high SNR and should guarantee that $I_6$ approaches $1$, then we find from~(\ref{equ_I6_one}) that $\delta^2 \gamma$ should go to $0$, which leads to
$\delta=(\tfrac{1}{\sqrt{\gamma}})^{\alpha}$, where $0 < \alpha < 1$. 
Finally, (\ref{equ_IXY_0_1}) is decomposed via this $\delta$ into
\begin{align}
I(X;Y) &=-\int_{\delta}^{1-\delta} p(y) \log(p(y)) ~dy -2 \int_0^{\delta} p(y) \log(p(y)) ~dy \nonumber \\
&+ \iint_{x,y \in [\delta, 1-\delta]} p(x) p(y|x) \log(p(y|x)) +
\iint_{x,y \notin [\delta, 1-\delta]} p(x) p(y|x) \log(p(y|x))\\
&= (1+\mathcal{O}(\tfrac{1}{\sqrt{\gamma}})) \cdot \frac{1}{2}\log\left( \frac{\gamma}{4\pi e}\right) ~+~\mathcal{O}(\delta).
\end{align} 
The cumbersome calculus details proving the last equality are not included for the sake of space. 
\end{proof}
The Gaussian differential entropy (\ref{equ_capa_approx}) is very close to $I(X;Y)$ above 2 bits per photon and becomes very accurate beyond 3 bits per photon 
where it coincides with the red upper envelope in Figure~\ref{fig_mutinfo_soft} at a high signal-to-noise ratio. The double variance $2\sigma^2$ in (\ref{equ_capa_approx}), originally found in (\ref{equ_p_y_x_capa}), comes from the superposition of the variances of $Z_1$ and $Z_2$ in the system model defined by (\ref{equ_XYmodel}). After canceling $U$, the model becomes $\tY=\tX+Z_1-Z_2$. $\tX$ and $Z_1$ are correlated, making the density expression relatively complicated when conditioning on $X$. In the small-noise regime, this correlation fades away, and the variance $2\sigma^2$ of the total additive Gaussian noise $Z_1-Z_2$ dominates the mutual information as in (\ref{equ_capa_approx}). At low and very low signal-to-noise ratios, one should use exact density expressions from Theorem~\ref{th_capa} and proceed via numerical integration to get exact values of $I(X;Y)$ and the corresponding SNR limits if the user accepts to apply a relatively low coding rate which is not the trend in TE-QKD where coding rates above $1/2$ are preferred which places us in the moderate and the high SNR region.

%%------------------------------------------------------------------
%%------------------------------------------------------------------
\section{Key-Reconciliation Codes\label{sec_codes}}
Following the complete characterization in Sections~\ref{sec_soft_output}-\ref{sec_capacity} of the time-entanglement QKD channel model described in Section~\ref{sec:Model}, we now introduce error-correcting codes to bring the error-rate performance as close as possible to the information theoretic limits corresponding to maximal achievable rates.

%%------------------------------------------------------------------
\subsection{Reed-Solomon Codes}
We consider the famous family of Reed-Solomon codes with an application to a frame of $N=2^m$ bins, i.e. $m$ coded bits per photon. In order to chose a high enough error-correction capacity, an RS code over $\F_q$ is considered, where $q$ is large enough. 
Each finite field element corresponds to $\log_2(q)/m$ photons. For simplicity, 
assume that $q=2^{\ell m}$, for some positive integer~$\ell$. 
The RS code has length $n=q-1$ (primitive) and dimension $k=n-2t$, 
so the targeted rate is $m \times \frac{k}{n}$ information bits per photon. 
One codeword of this $\mathcal{C}[n,k,t]_q$ RS code requires the transmission 
of a total of $n \times \log_2(q)/m$ photons to Alice and $n \times \log_2(q)/m$ photons to Bob, all with valid frames. 
After receiving the $n \times \log_2(q)/m$ valid frames, Alice converts the $n \times \log_2(q)$ bits received on the quantum channel into a length-$n$ word denoted by $c+e$, where $c \in \mathcal{C}$ and
$e \in \F_q^n$. Similarly, Bob converts his $n \times \log_2(q)$ received bits into $c+e'$,
where $e' \in \F_q^n$. In the hard-output channel model of Section~\ref{sec:Model}, 
$c+e$ is written at the input $\hX$ and $c+e'$ is read from the output $\hY$. 
On the public channel, Alice sends to Bob the syndrome $s=(c+e)H^t$, $s \in \F_q^{n-k}$, where $H$ is the parity-check matrix of $\mathcal{C}$. 
Given $s$ and given $c+e'$, the reconciliation performed by Bob is equivalent to finding Alice's word $c+e$. Bob proceeds as follows:
\begin{itemize}
\item Compute a syndrome $s'=(c+e')H^t$.
\item Feed $s'-s$ to an algebraic (Berlekamp-Massey \cite{Blahut2003}) decoder to find $e'-e$.
\item Subtract the error $e'-e$ from Bob's word to get $c+e'-(e'-e)=c+e$ the $\F_q^n$ word possessed by Alice. Replace all subtractions by additions in usual finite fields of characteristic $2$. 
\end{itemize}
The performance of RS $\mathcal{C}[n=63,k=43,t=10]_{q=64}$ code
is shown in Figure~\ref{fig_RS_perf}. One codeword requires the transmission
of a total of $126$ photons, where one field element carries two photons. 
The results show a large gain, e.g., about 58 dB of gain
for a bit error-rate $P_{eb}=10^{-5}$ after reconciliation. 

\begin{figure}[h]
\centering
\includegraphics[width=9cm,angle=270]{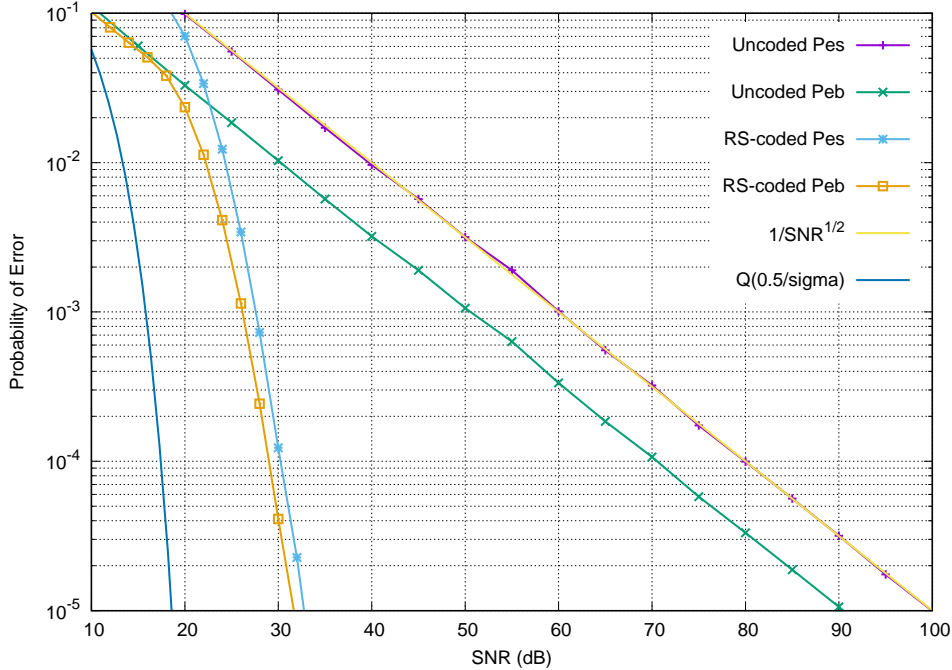}
\caption{Performance of the RS code $[63,43,t=10]_{64}$ on the hard-output time-entanglement QKD channel, for $N=8$ bins per frame, transmitting $2.05$ information bits per photon.}
\label{fig_RS_perf}
\end{figure}

The analysis of the algebraic decoder is easy thanks to its bounded-distance decoding in the Hamming space. A decoding error occurs each time the channel adds more than $t$ errors in $\F_q$. A simple union bound is obtained by summing from $t+1$ to $n$ errors. 
We proceed in the following steps to establish this bound for the RS code:\\
a) The uncoded symbol error probability over the TE-QKD channel is $P_e(\gamma)=\tfrac{2}{\sqrt{\pi}}(1-\tfrac{1}{N})\tfrac{1}{\sqrt{\gamma}}$.\\
b) For the RS code, the input probability of error per finite-field element is
$P_{in}(\gamma)=1-(1-P_e(\gamma))^{\ell}$.\\
c) The bound on the probability of error in $\F_q$ after decoding becomes
\begin{equation}
\label{equ_PeRS}
P_{eRS}(\gamma)=\sum_{i=t+1}^n \tfrac{i}{n}{n \choose i} P^i_{in}(\gamma) (1-P_{in}(\gamma))^{n-i}.
\end{equation}
d) The symbol (per photon) error probability after decoding is then 
$P_{eOut}(\gamma)=1-(1-P_{eRS}(\gamma))^{1/\ell}$.\\  
e) The probability of error per bit after Reed-Solomon decoding, 
given a Gray labeling of the bins, 
is well estimated by $P_{ebRS}(\gamma) = \tfrac{1}{\log_2(N)} P_{eOut}(\gamma)$.\\

The probability or error $P_{ebRS}$ obtained from (\ref{equ_PeRS}) perfectly fits the Monte Carlo method in the area where this method is tractable on a computer, i.e. for error rates in the interval $[10^{-7}, 10^{-1}]$. At $P_{ebRS}=10^{-10}$, the coding gain over the uncoded probability of error per bit is 158 dB! 
Such a huge gain is explained by the diversity order of the TE-QKD channel. 
The diversity order is defined as $\lim_{\gamma \rightarrow \infty} \tfrac{-\log(P_e)}{\log(\gamma)}$ \cite[Chapters~13-14]{Proakis2008}. From Proposition~\ref{prop_pe_uncoded} we know that the TE-QKD has a diversity order of $\tfrac{1}{2}$, it behaves like a half-diversity Nakagami fading channel. The error-correcting code increases the diversity order which is equivalent to increasing the slope of $P_e(\gamma)$. An additive Gaussian noise channel without fading has infinite diversity, with or without coding, making all curves look parallel. 
In presence of fading, a high diversity converts the channel into a Gaussian channel \cite{Boutros1998}. In practice, a diversity order beyond 8 could be barely distinguished from the local slope of $e^{-\gamma}$ on a Gaussian channel. In our case, from (\ref{equ_PeRS}), we deduce that the diversity order after algebraic RS decoding is $(t+1)/2$.
There is no asymptotic coding on the TE-QKD channel. The coding gain increases if measured at a lower probability of error.

%%------------------------------------------------------------------
\subsection{Binary BCH Codes \label{sec_bch}}
The TE-QKD channel does not generate error bursts. Errors are independent and the most common event is one erroneous bit per photon before decoding. In other words, the binary-burst error-correcting capability of Reed-Solomon codes is not exploited. Hence, we suggest to utilize a binary BCH code of the same binary length as the RS$[63,43]_{64}$, which is $63\times 6 = 378$ binary digits. We start from a primitive length of 511 and shorten down to 378. At $t=13$ the binary BCH code has a dimension $k=261$. This BCH$[378,261,t=13]_2$ code yields a diversity order $(t+1)/2=7$ better than the 5.5 order of the RS code shown in the previous section. The number of information bits per photon is $261/378\times 3=2.07$ bits for $N=8$ bins per frame.\\
Without adding any extra figure to this sub-section, the Monte Carlo simulation and the analytical bound show that the binary BCH$[378,261,t=13]$ code beats the RS$[63,43]_{64}$ code by 3 dB in signal-to-noise ratio at $P_{ebRS}=P_{ebBCH}=10^{-5}$. To get the coding gain at a lower probability of error, we propose the following very tight union bound:\\
a) The uncoded symbol error probability over the TE-QKD channel is $P_e(\gamma)=\tfrac{2}{\sqrt{\pi}}(1-\tfrac{1}{N})\tfrac{1}{\sqrt{\gamma}}$. Below $10^{-1}$ a maximum of one bit error occurs in a block of $m=\log_2(N)$ coded bits thanks to Gray labeling. There are $n/m$ such blocks per BCH codeword involving individual binary errors.\\
b) The bound on the probability of error in $\F_2$ after BCH decoding becomes
\begin{equation}
\label{equ_PebBCH}
P_{ebBCH}(\gamma)=\sum_{i=t+1}^{n/m} \tfrac{i}{n}{n/m \choose i} P^i_e(\gamma) (1-P_e(\gamma))^{n/m-i}.
\end{equation}
At $P_{ebRS}=P_{ebBCH}=10^{-10}$, the binary BCH$[378,261,t=13]$ code beats the RS$[63,43]_{64}$ code by 5~dB. This value corresponds to a 163~dB of BCH coding gain with respect to the uncoded photons at $N=8$ bins per frame. 
Notice that the reconciliation at Bob's side for BCH codes (binary or non-binary) is identical to the reconciliation described in the previous section for Reed-Solomon codes where the syndrome $s'-s$ is fed to a Berlekamp-Massey decoder. 
%%------------------------------------------------------------------
\subsection{Graph-Based LDPC Codes \label{sec_ldpc}}
The big impact of LDPC codes on the performance of polarization-based QKD systems was already demonstrated in~\cite{Elkouss2009} for the reconciliation of discrete random variables, with a BSC channel model. Low-density parity-check codes \cite{Gallager1963}\cite{Richardson2008} are very flexible in terms of length, coding rate, and decoding methods. As usual, the LDPC code parity-check matrix is the adjacency matrix of a bipartite Tanner graph with $n$ variable nodes and $n-k$ check nodes, assuming that the graph is $(d_v, d_c)$-regular. For finite fields $\F_q$ with $q>2$, non-zero elements of the adjacency matrix are replaced by elements from $\F_q \setminus \{0\}$. The standard method for decoding LDPC codes is belief propagation (BP), i.e. iterative probabilistic decoding. Codes over a large field $F_q$ or a large ring $\Z/q\Z$ could be considered \cite{Boutros2018} in order to minimize the loss during the symbol-to-bit soft values conversion. It is also possible to use joint local-global LDPC codes with optimized bin mapping to achieve good performance \cite{codes:YangSCWD19} or apply multilevel-coding as in \cite{codes:ZhouWW13}, although these papers consider a different QKD channel model. In this paper, we show the impact of LDPC codes on TE-QKD with a $(3,9)$-regular binary LDPC code only. The coding rate is $2/3$ guaranteeing $2$ exchanged bits per photon when the frame has 8 bins, however we consider a short length $n=384$ ($64 \times 6$) comparable to the RS and BCH codes given in the previous sub-sections, and a longer code with $n=9999$ to illustrate a performance close enough to Shannon limit.\\

The symbol/bin APP is found via (\ref{equ_app}), where $APP(i)=APP(\hX=i)$ is the {\em a posteriori} probability of bin number $i$, $i \in \Z_N$. Then the APP of binary digit $b_{\ell}$, where $\ell \in \Z_m$, $m=\log_2(N)$, is derived by the following marginalization 
\begin{equation}
\label{equ_binary_app}
APP(b_{\ell}) = \sum_{i \in \Z_N ~:~ b_{\ell}} APP(\hX=i).
\end{equation}
The above marginalization depends on the type of binary labeling. Our paper is restricted to $N$ bins per frame with a Gray labeling of $\log_2(N)$ bits per bin. Figure~\ref{fig_ldpc_perf} shows the bit error-rate versus $\gamma$ for the binary LDPC 
code on the TE-QKD soft-output channel at $n=384$ bits and $n=9999$ bits. 
They respectively gain 12 and 16 dB with respect to the BCH$[378,261]$ code, at a bit error 
probability of~$10^{-5}$. If compared to the uncoded TE-QKD, the coding gain is 73 dB and 77 dB respectively. At length $n=9999$, the LDPC code is on top of the Shannon limit for a TE-QKD hard-output channel ($\gamma_{limit}=12.61$ dB) and is 2 dB only from the Shannon limit of the soft-output TE-QKD channel ($\gamma_{limit}=10.45$ dB). We see no reason for using longer LDPC codes to catch an extra 1-2 dB given that the total coding gain with respect to the no-coding case already equals 77 dB!

\begin{figure}[h]
\centering
\includegraphics[width=9cm,angle=270]{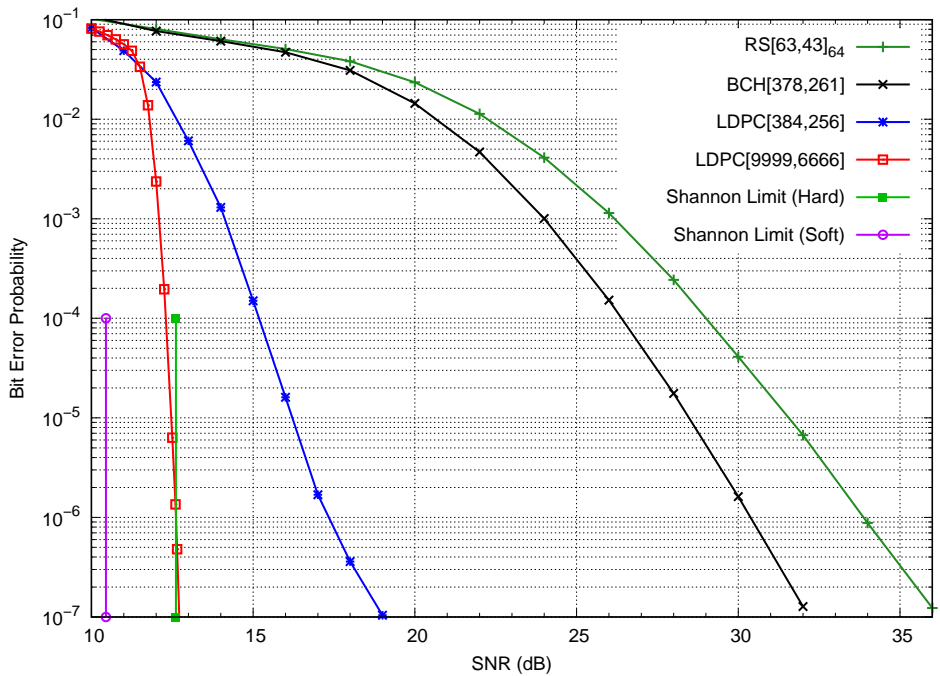}
\caption{Performance of the $(3,9)$-regular binary LDPC code at length $n=384$ bits and $n=9999$ bits on the soft-output time-entanglement QKD channel, for $N=8$ bins per frame, transmitting $2.0$ information bits per photon.}
\label{fig_ldpc_perf}
\end{figure}

In practice, if a lab system implementation requires a less complex expression for $APP(\hX=i)$ without the erfc()/Q() function and without integration, (\ref{equ_app}) can be simplified by assuming that the Gaussian density has the effect of a Dirac impulse at small $\sigma$ and using the $\propto$ symbol (proportional to) since the denominator does not depend on the index $i$, we get:
\begin{align*}
APP(i) &\propto
\int_0^N \frac{1}{\sqrt{2\pi\sigma^2}} e^{-\frac{(y-u)^2}{2\sigma^2}}\cdot \left[Q\left(\frac{i-u}{\sigma}\right)-Q\left(\frac{i+1-u}{\sigma}\right)\right] \,du \\
& \propto \left[ Q\left(\frac{i-y}{\sigma}\right)-Q\left(\frac{i+1-y}{\sigma}\right) \right].
\end{align*}
Then, depending on the sign of the arguments $i-y$ and $i+1-y$, we approximate $Q(x)$ by 
$\tfrac{1}{2}e^{-\tfrac{x^2}{2}}$ (if $x\ge 0$) and by $1-\tfrac{1}{2}e^{-\tfrac{x^2}{2}}$ (if $x < 0$). Let $j=\lfloor Y \rfloor$ be the bin position of $Y$ on Bob's side, i.e. $Y \in [j, j+1[$. The simplified APP expressions become:
\begin{align}
\text{If}~i=j,~~APP(i) &\propto \left[1-\tfrac{1}{2}e^{-\tfrac{(y-i)^2}{2\sigma^2}}
-\tfrac{1}{2}e^{-\tfrac{(y-i-1)^2}{2\sigma^2}}\right], \label{equ_app1_approx}\\
\text{If}~i\ne j,~~APP(i) &\propto sign(j-i) \cdot \left[\tfrac{1}{2}e^{-\tfrac{(y-i-1)^2}{2\sigma^2}}
-\tfrac{1}{2}e^{-\tfrac{(y-i)^2}{2\sigma^2}}\right]. \label{equ_app2_approx}
\end{align}
When (\ref{equ_app1_approx})-(\ref{equ_app2_approx}) are utilized in the BP decoder of the binary LDPC code over the TE-QKD soft-output channel, the loss is limited to 0.25-0.30 dB with respect to the exact expression (\ref{equ_app}). This is a minuscule loss when dealing with coding gains above 50 dB.

Notice that we are not showing a performance of the LDPC code over a hard-output channel. Indeed, optimal BP decoding is identical whether the channel output is soft or not, i.e., the BP decoder is the same decoder on both a Gaussian-like channel and a BSC-like channel. 
The gap between hard and soft is about 8.5 dB for the LDPC$[384,256]$ and about 4 dB
for the LDPC$[9999,6666]$ at a bit error rate of $10^{-5}$. Suppose the system implementation possesses an optimal BP decoder, but the exact photon position is unavailable; only the bin number is available. In such a case, the lab implementation is forced to use LDPC codes on a hard-output channel, and the binary digits APP expression (\ref{equ_binary_app}) becomes
\begin{equation}
APP(b_{\ell}) \propto \sum_{i \in \Z_N ~:~ b_{\ell}} \hpi_i \times p_{i,j},
\end{equation}
where $j=\lfloor Y \rfloor$, $\hpi_i$ is given by (\ref{equ_hpi}) or (\ref{equ_hpi0_small_sigma})-(\ref{equ_hpii_small_sigma}) at small $\sigma$, and $p_{i,j}$ is given by (\ref{equ_pij}) or (\ref{equ_p01_small_sigma})-(\ref{equ_p12_small_sigma}) in the small $\sigma$ regime. Coding theorists and practitioners could also use convolutional codes, turbo codes, polar codes, and other binary or non-binary algebraic codes with short or moderate length to achieve large coding gains on the TE-QKD channel.  
%%------------------------------------------------------------------
\subsection{A summary of capacity limits at different frame sizes}
We complete the current section by a table summarizing important information theoretical limits on the time-entanglement QKD channel, with both hard and soft output. 
Shannon limit in terms of SNR is the value of the non-normalized signal-to-noise ratio $\gamma$ such that mutual information is equal to the targeted information exchange rate, $I(\hX;\hY)=\tfrac{k}{n} \log_2(N)$ for a hard output and $I(\hX;Y)=\tfrac{k}{n} \log_2(N)$ for a soft output. Table~\ref{tab_limits} has seven columns with parameters covering 8 bins per frame up to 64 bins per frame. The last two rows correspond to SNR and standard deviation values achieved by the BCH and the LDPC codes as found in sub-sections~\ref{sec_bch} and~\ref{sec_ldpc}. 

\begin{table}[!h]
    \centering
      \caption{Information theoretical (Shannon) limits for TE-QKD.}
    \label{tab_limits}
    \begin{tabular}{|c|c|c|r|l|r|l|} \hline
      N   &  Bits & $R=k/n$ & SNR limit & $\sigma/N$ & SNR limit & $\sigma/N$ \\
     bins per frame & per photon & code rate & hard &  hard & soft & soft \\ \hline \hline
     8	&2.0	&2/3	&12.61 dB	&0.029269	&10.45 dB	&0.037533 \\ \hline
     16	&3.0	&3/4	&13.29 dB	&0.013532	&10.85 dB	&0.017922 \\ \hline
     32	&3.0	&3/5	&3.88 dB	&0.019992	&3.46 dB	&0.020982 \\ \hline
     32	&4.0	&4/5	&13.61 dB	&0.0065215	&11.04 dB	&0.0087670 \\ \hline
     64	&4.0	&2/3	&4.01 dB	&0.0098474	&3.58 dB	&0.010347  \\ \hline
     64	&5.0	&5/6	&13.77 dB	&0.0032012	&11.13 dB	&0.0043383  \\ \hline \hline
     8   &2.0	&2/3	&28.49 dB   &0.0047034 &           &            \\
     BCH, n=378 & & & $P_{eb}=10^{-5}$ & achieved & & \\ \hline
     8   &2.0	&2/3	&           &           &12.47 dB   &0.029745   \\
     LDPC, n=9999 & & & & & $P_{eb}=10^{-5}$ & achieved  \\ \hline
     \end{tabular}
\end{table}

The signal-noise ratio soft-decoding limits listed in Table~\ref{tab_limits} appear to be close to two values, one SNR around 10-11 dB and a lower SNR around 3.5 dB. 
The hard-decoding limits are higher than soft-decoding limits, because $I(\hX;\hY) \le I(\hX;Y)$, the gap depends on the frame size and the coding rate. Of course, the hard-soft gap vanishes at small coding rates (below $1/2$) and increases at high coding rates when 
mutual information approaches the asymptote $\log_2(N)$.\\
The two typical values of soft-decoding SNR limits are explained or interpreted for small $\sigma$ via (\ref{equ_capa_approx}):
\begin{equation}
\label{equ_gamma_and_backoff}
\frac{1}{2}\log\left( \frac{\bgamma}{4\pi e}\right)  = \log_2(N)-b, 
\end{equation}
where $\bgamma=N^2\gamma$ and $b$ is a backoff value. Here, $b=1$ bit or $b=2$ bits in Table~\ref{tab_limits}. Then, solving (\ref{equ_gamma_and_backoff}) yields $\gamma=(4\pi e)/2^{2b}$. We get $\gamma=9.31$ dB for $b=1$ and $\gamma=3.29$ for $b=2$.
The difference with the values in the 6th column of Table~\ref{tab_limits} is due to $I(\hX;Y)$ going away from the envelope $I(X;Y)$ to follow its own asymptote. 
We hope that SNR limits given in Table~\ref{tab_limits} will be useful to physicists and  coding theorists working in this QKD field. 
%%~\\
%%Non-binary LDPC with soft input? \\
%%~\\
%%------------------------------------------------------------------
%%------------------------------------------------------------------
\section{Conclusions}
We focused on the time entanglement-based QKD when the photon arrival detectors suffer from time jitter. We presented a rigorous analysis 
of secret key information rates and proposed and tested several codes for information reconciliation to approach the maximum secret key rates. These achievable secret key rates are much higher than the maximum achievable by polarization entanglement-based QKD. However, practical photon detectors suffer from other impairments, e.g., dark currents and downtime, which may cause further rate loss. These impairments should be a subject of future work.
%%------------------------------------------------------------------
\bibliographystyle{ieeetr}
\bibliography{bibio}

\end{document}